\documentclass[journal, final]{IEEEtran}
\usepackage[T1]{fontenc}
\usepackage[utf8]{inputenc}
\usepackage[english]{babel}
\usepackage{balance}
\usepackage{cite}
\usepackage{amsmath, mathtools}
\usepackage{amsfonts,amsthm,bm}
\usepackage{amssymb}
\usepackage{comment}
\usepackage{graphicx}
\usepackage{standalone}
\usepackage{dsfont}
\usepackage{mathtools}
\usepackage{color, colortbl}
\usepackage{soul}
\usepackage[normalem]{ulem}
\usepackage[acronym,shortcuts]{glossaries}
\usepackage{tikz}
\usepackage{arydshln}
\usepackage{bbm}
\usepackage{svg} 
\usepackage{algorithm,algorithmic}

\usepackage{tikz}
\usepackage{pgfplots}
\pgfplotsset{compat=newest}
\pgfplotsset{plot coordinates/math parser=false}
\newlength\fheight
\newlength\fwidth
\usetikzlibrary{plotmarks,patterns,decorations.pathreplacing,backgrounds,calc,arrows,arrows.meta,spy,matrix}
\usepgfplotslibrary{patchplots,groupplots}
\usepackage{soul,color}
\usepackage{tikzscale}
\usetikzlibrary {patterns.meta}

\DeclareMathOperator*{\argmax}{\arg\!\max}

\newtheorem{theorem}{Theorem}

\newacronym{5g}{5G}{fifth-generation}
\newacronym{b5g}{B5G}{beyond-fifth-generation}
\newacronym{6g}{6G}{sixth-generation}
\newacronym{ai}{AI}{activation information}
\newacronym{amp}{AMP}{approximate message passing}
\newacronym{aoi}{AoI}{age-of-information}
\newacronym{awgn}{AWGN}{additive white Gaussian noise}
\newacronym{bs}{BS}{base station}
\newacronym{eccdf}{ECCDF}{empirical complementary cumulative distribution function}
\newacronym{cs}{CS}{compressed sensing}
\newacronym{cra}{CRA}{compressive \ac{ra}}
\newacronym{csi}{CSI}{channel state information}
\newacronym{dt}{DT}{data transmission}
\newacronym{fifo}{FIFO}{first in, first out}
\newacronym{fsa}{FSA}{framed slotted-ALOHA}
\newacronym{gf}{GF}{grant-free}
\newacronym{gfeo}{GFEO}{greedy frame efficiency optimizer}
\newacronym{gmt}{GMT}{greedy maximum throughput}
\newacronym{harq}{HARQ}{hybrid automatic repeat request}
\newacronym{hmm}{HMM}{hidden Markov model}
\newacronym{iiot}{IIoT}{industrial internet-of-things}
\newacronym{iot}{IoT}{internet-of-things}
\newacronym{isa}{ISA}{i.i.d. slot allocation}
\newacronym{kpi}{KPI}{key performance indicator}
\newacronym{lsfc}{LSFC}{large-scale fading coefficient}
\newacronym{lri}{LRI}{linear-reward-inaction}
\newacronym{m2m}{M2M}{machine-to-machine}
\newacronym{map}{MAP}{maximum a posteriori probability}
\newacronym{mdp}{MDP}{markov decision process}
\newacronym{mimo}{MIMO}{multiple input-multiple output}
\newacronym{minlp}{MINLP}{mixed integer non-linear programming}
\newacronym{ml}{ML}{maximum likelihood}
\newacronym{mmpc}{MMPC}{min-max pairwise correlation}
\newacronym{mmtc}{mMTC}{massive \ac{mtc}}
\newacronym{mra}{MRA}{massive random access}
\newacronym{mtc}{MTC}{machine-type communications}
\newacronym{mtd}{MTD}{machine-type device}
\newacronym{nr}{NR}{new radio}
\newacronym{noma}{NOMA}{non-orthogonal multiple access}
\newacronym{oma}{OMA}{orthogonal multiple access}
\newacronym{omp}{OMP}{orthogonal matching pursuit}
\newacronym{pam}{PAM}{pulse-amplitude modulation}
\newacronym{pb}{PB}{periodic beacon}
\newacronym{pdf}{PDF}{probability density function}
\newacronym{pia}{PIA}{partial information acquisition}
\newacronym{pima}{PIMA}{partial-information multiple access}
\newacronym{pomdp}{PO-MDP}{partially observable-Markov decision process}
\newacronym{ra}{RA}{random access}
\newacronym{rb}{RB}{reservation beacon}
\newacronym{rfid}{RFID}{radio-frequency identification}
\newacronym{rl}{RL}{reinforcement learning}
\newacronym{saloha}{SALOHA}{slotted ALOHA}
\newacronym{sb}{SB}{scheduling beacon}
\newacronym{snr}{SNR}{signal-to-noise ratio}
\newacronym{tdma}{TDMA}{time-division multiple-access}
\newacronym{ura}{URA}{unsourced random access}
\newacronym{usd}{usd}{unit symbol duration}
\newacronym{urllc}{URLLC}{ultra-reliable low-latency communications}

\def\BibTeX{{\rm B\kern-.05em{\sc i\kern-.025em b}\kern-.08em
    T\kern-.1667em\lower.7ex\hbox{E}\kern-.125emX}}
\AtBeginDocument{\definecolor{ojcolor}{cmyk}{0.93,0.59,0.15,0.02}}

\begin{document}

\title{Semi-Grant-Free Orthogonal Multiple Access\\with Partial-Information\\for Short Packet Transmissions}
\author{Alberto Rech\IEEEauthorrefmark{1}\IEEEauthorrefmark{3}, Stefano Tomasin\IEEEauthorrefmark{1}\IEEEauthorrefmark{2}, Lorenzo Vangelista\IEEEauthorrefmark{1}, and Cristina Costa\IEEEauthorrefmark{4} \\ 
\IEEEauthorrefmark{1}Department of Information Engineering, University of Padova, Italy.\\
\IEEEauthorrefmark{2}Department of Mathematics, University of Padova, Italy.\\
\IEEEauthorrefmark{3}Smart Networks and Services, Fondazione Bruno Kessler, Trento, Italy.\\
\IEEEauthorrefmark{4}S2N National Lab, CNIT, Genoa, Italy.\\ 
\small
\texttt{alberto.rech.2@phd.unipd.it, stefano.tomasin@unipd.it,} \\
\texttt{lorenzo.vangelista@unipd.it, cristina.costa@cnit.it}
\thanks{Part of this paper has been presented at the International Conference on Ubiquitous and Future Networks (ICUFN) \cite{rech2023partial}.}}

\maketitle

\begin{abstract}
Next-generation \ac{iot} networks require extremely low latency, complexity, and collision probability.  
We introduce the novel \ac{pima} scheme, a semi-\ac{gf} coordinated \ac{ra} protocol for short packet transmission, with the aim of reducing the latency and packet loss of traditional multiple access schemes, as well as more recent preamble-based schemes. 
With \ac{pima}, the \ac{bs} acquires partial information on instantaneous traffic conditions in the \ac{pia} sub-frame, estimating the number of active devices, i.e., having packets waiting for transmission in their queue.
Based on this estimate, the \ac{bs} chooses both the total number of slots to be allocated in the \ac{dt} sub-frame and the respective user-to-slot assignment.
Although collisions may still occur due to multiple users assigned to the same slot, they are drastically reduced with respect to the \ac{saloha} scheme, while achieving lower latency than both \ac{tdma} and preamble-based protocols, due to the extremely reduced overhead of the \ac{pia} sub-frame.
Finally, we analyze and assess the performance of \ac{pima} under various activation statistics, proving the robustness of the proposed solution to the intensity of traffic, also with burst traffic.
\end{abstract}

\begin{picture}(0,0)(0,-480)
\put(0,0){
\put(0,0){\qquad \qquad \quad This paper has been submitted to IEEE for publication. Copyright may change without notice.}}
\end{picture}

\glsresetall

\glsresetall

\section{Introduction}\label{sec:introduction}
\Ac{mtc} is a use case of \ac{b5g} cellular networks with increasing relevance, as the focus of mobile communications shifts from humans to machines. To support this scenario, several new technological solutions should be adopted, including specific multiple access schemes for both flavors of \ac{mtc}, i.e., \ac{mmtc} and \ac{urllc}. Both show several distinct features and challenges to meet the requirements of emerging \ac{iot} applications and services.  For example, \ac{urllc} use cases target a maximum latency of 1~ms and reliability of 99.99999\%  (e.g. mission-critical applications), while \ac{mmtc} scenarios require supporting devices with density up to 1~million devices per km$^2$ (e.g., for industrial \ac{iot}). 
The focus of this paper is the medium access control of uplink transmissions in an \ac{mtc} scenario, wherein users (or \acp{mtd}) transmit to a common \ac{bs}. 
Although access procedures based on resource requests and grants have been discussed for \ac{mtc} \cite{Centenaro17}, \ac{mmtc} \cite{Riolo21}, and \ac{urllc} \cite{Popovski19}, the sporadic nature of transmissions by a large number of users makes these schemes inefficient, and a \ac{gf} or semi-\ac{gf} solution is to be preferred. 
\Ac{gf} approaches address this issue since users can transmit data immediately, without waiting for approval from \ac{bs}. \ac{gf} approaches include several different techniques, which can be classified into uncoordinated and coordinated \ac{ra} \cite{ElTanab21}.
A survey of the literature related to these alternative schemes is provided in Section~\ref{litrev}. In any case, existing multiple access solutions are either based on the exact knowledge by the \ac{bs} of which users are going to transmit or on no knowledge of the users' state (i.e., if they have a packet to transmit or not).

In \cite{rech2023partial}, we introduced a novel approach to multiple access, in which \ac{bs} first acquire a {\em partial information} on the state of the user, and then schedules the transmissions. In particular, first, the \ac{bs} estimates the {\em number of users} with packets to transmit, without knowing their identities. Then, based only on this information, allocates users to slots. The allocation is performed in a non-exclusive manner, i.e., a slot is typically allocated to multiple users whose transmissions may collide. Indeed, without knowing the {\em identity} of users with packets to transmit ({\rm active users}), collisions could only be avoided using \ac{tdma}, which is however extremely inefficient for sporadic traffic. 

The resulting scheme is denoted as \ac{pima}. In \ac{pima}, time is organized in frames of {\em variable length}, each split into two sub-frames. The first is the \ac{pia} sub-frame, where all active users simultaneously send a signal to the \ac{bs}. The \ac{bs} performs a \ac{map} estimation of the number of active users. Based on this knowledge, \ac{bs} then assigns one slot to each user in the system for transmission in the \ac{dt} subframe.

However, the \ac{pia} sub-frame in the scheme of \cite{rech2023partial} introduced a significant overhead, as the users must transmit a long random sequence and the \ac{bs} estimated the total received power to count the active users. Moreover, an in-depth analysis of the scheme was missing, considering the differentiated scenarios of next-generation cellular networks. In this paper, we partially re-design \ac{pima} and provide an in-depth analysis of its performance in various scenarios. In particular, the main contributions of this paper are the following.
\begin{enumerate}
\item We focus on the short packet transmission scenario, with packets consisting of a few complex symbols each and we improve the user enumeration procedure. By assuming \ac{csi} is perfectly acquired at each user with a downlink beacon, users can precode their transmission so that they coherently combine at the \ac{bs}, similarly to a \ac{pam} signal. Therefore the \ac{bs} can reliably count the active users with a significantly shorter \ac{pia} sub-frame, reducing the overhead.
\item We consider three different packet generation statistics, including a bursty traffic generation in the case of \ac{mra}, wherein the number of users in the system is arbitrarily large while the number of active users remains finite;
\item We prove that with independent identically distributed (i.i.d.) activation, allocating each user to a single slot maximizes the resulting frame efficiency.
\item We prove that with \ac{mra} traffic, the maximum frame efficiency is obtained by allocating a number of slots equal to the number of active users.
\item We significantly extend the numerical evaluation part, comparing \ac{pima} also with state-of-the-art preamble-based approaches.
\end{enumerate}

The remainder of the paper is organized as follows. Section \ref{litrev} provides a review of the literature of multiple access schemes. In Section~\ref{sec:systemmodel} and Section~\ref{sec:protocol}, we first introduce the system model and the packet generation processes, then we describe the frame structure and the \ac{pima} protocol. Section~\ref{sec:numest} provides details of the user enumeration task for partial information acquisition. 
Then, the optimal scheduler for i.i.d. activations is derived in Section~\ref{sec:frameeff}.
In Section~\ref{sec:numericalresults} we discuss the numerical results and compare \ac{pima} with the state-of-the-art orthogonal schedulers. schedulers. Finally, Section~\ref{sec:conclusions} draws some conclusions.

\section{Literature Review}\label{litrev}
Multiple access schemes can broadly categorized into uncoordinated \ac{ra}, coordinated \ac{ra}, and fast uplink grants.

\vspace{5pt}\noindent\emph{Uncoordinated \ac{ra}.} In uncoordinated \ac{ra}, users transmit at random time instants, and specific techniques are adopted at the receiver to mitigate the effects of collisions. Such techniques include traditional \ac{oma}, which is outperformed by \ac{noma} schemes \cite{Saito13Non}, for which the symbols of each user are not allocated to orthogonal resources. An effort has been made towards the unification of the different \ac{noma} solutions and the performance gains with respect to their orthogonal counterparts have been highlighted \cite{Meng21Advanced, Chen18Toward}. However, \ac{noma} requires advanced pairing and power allocation techniques, as well as powerful channel coding and interference cancelation mechanisms that only partially mitigate collision effects \cite{Dai18}. Under these conditions, \ac{bs} may become prohibitively complex to serve a large number of users.
In recent years, unsourced  \ac{ra} has been proposed to manage a massive number of devices \cite{Polyanskiy17}: at any time, a fraction of devices transmit simultaneously using the same channel codebook. The receiver decodes the arriving messages without knowing the identities of the transmitters. Although this approach is very effective in managing many users, good performance can be achieved only for very small payloads and with high-complexity massive \ac{mimo} receivers \cite{fengler2021non, decurninge2020tensor, rech2023unsourced}. Moreover, most works assume perfect knowledge of the number of active users, and only recently an information theoretical analysis of \ac{ura} with random user activity has been proposed \cite{Ngo2023Unsourced}.

\vspace{5pt}\noindent\emph{Coordinated \ac{ra}.}  These solutions typically divide the time into slots, each with the duration of one packet. \Ac{saloha} is the simplest and most widely adopted coordinated \ac{ra} protocol: users transmit at the beginning of the first slot available after packet generation. When collisions occur,  users randomly delay the retransmissions of collided packets. A reduction of collisions is achieved by dividing time into frames; each of these is split into slots wherein users transmit at random, in what is known as \ac{fsa}, widely adopted in \ac{rfid} systems \cite{Lee05, Su16}. Nevertheless, coordinated \ac{ra} solutions are particularly useful when user activations are highly correlated, for example, as a result of correlated underlying traffic generation \cite{3GPP37868}.
Also, re-transmissions (with the consequent accumulation of packets in user queues) introduce a correlation of transmissions among users:  this further increases collisions, while it can also be exploited to indirectly coordinate \ac{ra}. In the literature, correlation-based schedulers have recently gained attention as a possible breakthrough for multiple access in \ac{mtc}. Such schemes typically rely on the knowledge of traffic generation statistics \cite{Kalor18}, or learn the traffic correlation by exploiting the capabilities of \acp{hmm} tracking successes and collisions in previous frames \cite{Moretto21a}, or reinforcement learning techniques \cite{Rech21, Destounis19}.
In another line of research, preamble-based random access techniques are more suitable in the \ac{iot} context. For example, with \textit{multichannel ALOHA}, active users choose a preamble from the common preamble pool, that is, a codebook of orthogonal preambles. Then, the \ac{bs} can effectively differentiate between multiple users attempting to access the network simultaneously and schedule the users for data transmission in a second dedicated sub-frame, which typically has a fixed length. However, such schemes require different preamble lengths based on the number of active users to perform optimally. Indeed, while long preambles support more active numbers, their entailed communication overhead grows rapidly.
Non-orthogonal preambles to reduce overhead have been studied in the context of \ac{cra} \cite{Wunder14Compressive, Seo19Low, Choi20On}; with \ac{cra} the \ac{bs} resorts to \ac{cs} to identify active users \cite{Donoho06Compressed}. However, multiple measurements are needed to accurately estimate the preambles, therefore requiring \ac{mimo} receivers or multiple transmission steps.

\vspace{5pt}\noindent\emph{Fast Uplink Grant.}  In fast uplink grant schemes \cite{Ali19} the \ac{bs} schedules one slot for each user, without any resource request, so access randomness is removed, while coordination remains. Slots are usually shared by multiple users, and collisions may still occur. The exploitation of traffic correlation has also been considered in the design of fast-uplink grant protocols, involving multiarmed bandits-based methods \cite{Ali18Sleeping} and machine learning tools for traffic prediction \cite{Shehab20Traffic}.

\vspace{5pt}\noindent\emph{\ac{pima} Protocol Classification.} The \ac{pima} protocol proposed in this paper falls into the category of the coordinated \ac{ra} schemes, and in particular, can be considered a semi-grant-free \ac{ra} solution. Indeed, with respect to other two-step \ac{ra}-access schemes consisting of preamble and data transmission stages, \ac{pima} acquires only partial information on the activation statistics in the \ac{pia} sub-frame, avoiding revealing the users' identities. Likewise, \ac{pima} cannot be considered a fast uplink grant approach, due to its partial information acquisition at the \ac{bs}. 

\vspace{5pt}\noindent\emph{Notation.} Scalars are indicated in italic letters; vectors and matrices are indicated in boldface lowercase and uppercase letters, respectively. Sets are denoted by calligraphic uppercase letters and $|\mathcal{A}|$ denotes the cardinality of the set $\mathcal{A}$. $\mathbb{P}(\cdot)$ denotes the probability operator, $\mathbb{E}[\cdot]$ denotes the statistical expectation, and $\log(\cdot)$ indicates the natural logarithm function. 

\section{System Model}\label{sec:systemmodel}
We consider the uplink of a multiple access scenario with $N$ single-antenna users transmitting to a common single-antenna \ac{bs}. We assume that the value of $N$ is known at the \ac{bs}.

\vspace{5pt}\noindent\emph{Time Organization.} Time is split into {\em frames}, each comprising an integer number of {\em slots} and an additional short time interval, whose purpose is described in the following. Each slot has a fixed duration of $T_{\rm s}$ complex symbols, while each frame comprises a different number of slots. Perfect time synchronization at the \ac{bs} is assumed, thus each user can transmit signals with specific times of arrival at the \ac{bs}.

Each user may transmit at most one packet per frame, each of the duration of one slot, and each user transmits at most one packet per frame. In the following, $t$ denotes the frame index. We consider a multiple access protocol, where the same slot is in general assigned to multiple users for transmission.

\vspace{5pt}\noindent\emph{Channel.}
Due to the scheduling of the same slot to multiple users, collisions between packets may occur. We assume that when two or more users transmit in the same slot, a collision occurs, preventing the decoding of all collided packets by the \ac{bs}.
Successful transmissions are acknowledged by the \ac{bs} at the beginning of the following frame, and upon collision, collided users retransmit their packets in the following frame. We assume that the \ac{bs} always correctly decodes the received packets in slots without collision, thus the channel does not introduce other sources of communication errors.

\subsection{Packet Generation and Buffering}
Packets generated in frame $t$ by user $n$ are stored in its buffer and transmitted according to a \ac{fifo} policy.
In the following, we consider both the case of finite and infinite buffer capacity. In the case of finite-length queues, to ensure data freshness, whenever a new packet is generated while the buffer is at full capacity, the oldest packet is dropped.
Let $B_n$ be the queue length of the user $n$ queue at the beginning of frame $t$.
If $B_n(t) > 0$, the buffer of user $n$ is non-empty, and $n$ is said to be \textit{active}, instead, if $B_n(t)= 0$, its queue is empty and user $n$ is considered \textit{inactive}. The total number of active users at the beginning of frame $t$ is $\nu(t)$.

\subsection{Activations Statistics}\label{sec:actstat}

We analyze the performance of \ac{pima} in three different scenarios, depending on the users' activation statistics.

\vspace{5pt}\noindent\emph{i.i.d. Activations.}
The activation statistics of the users are i.i.d. when a) users have a queue for one packet only, and b) all colliding packets are dropped at the receiver after the first transmission. Such a scenario is typical of monitoring systems requiring frequent updates and data freshness \cite{munari22} and will be discussed in detail in Section~\ref{sec:iidactivations}.

\vspace{5pt}\noindent\emph{Correlated Activations.}
In this scenario, retransmissions of previously collided packets are allowed. 
Due to the presence of queues and retransmissions, users' activations are generally correlated, since users colliding at frame $t$ will deterministically retransmit in the following frames. In this case, we will assume to have queues of infinite length, and the stability condition of all queues is assumed.

For both the i.i.d. and the correlated activations cases, we assume that at each user the traffic generation, also denoted as {\em packet arrival process}, follows a Poisson distribution with parameter $\lambda$. The well-known properties of Poisson processes provide a total normalized arrival rate of $\Lambda_{T} = N\lambda/T_{\rm s}$, and when $\Lambda_{T} = 1$, on average one packet is generated during the duration $T_{\rm s}$ of a slot.

\vspace{5pt}\noindent\emph{Bursty Activations.}
In this scenario, we assume a bursty traffic model, in which a finite subset of users activates at the same time, with each user generating a single packet to transmit. Retransmissions of colliding packets are also allowed in this case, while the buffer capacity is unitary for all users.    
In particular, we assume that at each burst the number of generations of packets follows a Poisson distribution with parameter $\Lambda_{\rm B}$, and we denote by $\tau_{\rm B}$ the random variable of the interarrival time between consecutive bursts. 

\section{Partial Information-Multiple Access }\label{sec:protocol}
\begin{figure*}
    \centering
	\includegraphics[width = \textwidth]{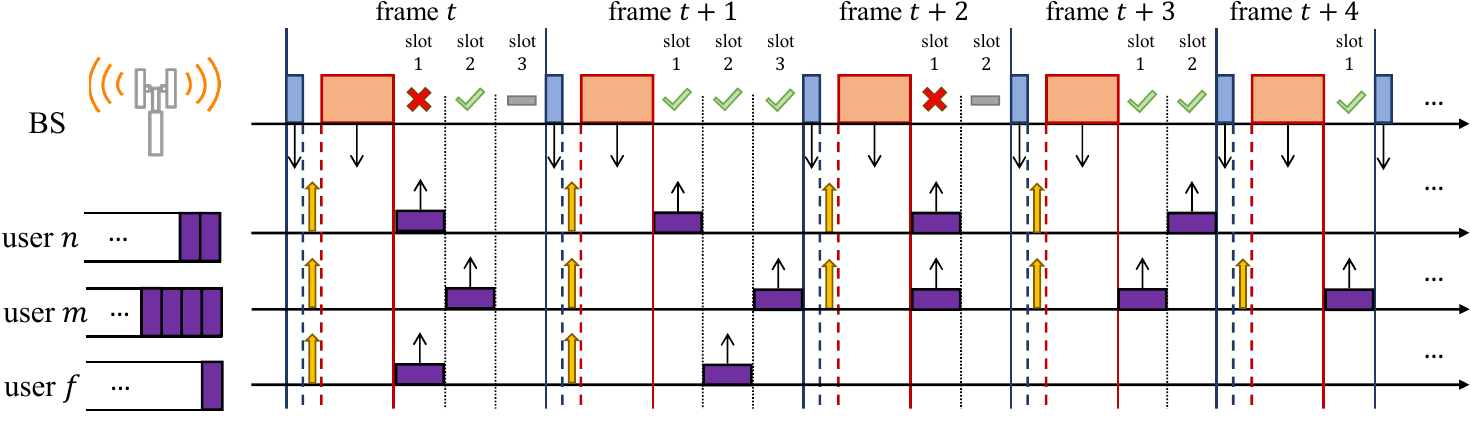}
	\caption{Example of the \ac{pima} protocol and its frame structure. In this example, the user $n$, $m$, and $f$ are active at the beginning of frame $t$, each with a different number of packets to transmit (purple rectangles). The \ac{rb} and the \ac{sb} transmitted by the \ac{bs} in the downlink are, respectively, represented by the blue and orange rectangles, while the yellow arrow represents the \ac{ai} signal for the user enumeration. For drawing simplicity, in this example no packets are generated after the beginning of frame $t$.}
	\label{fig:pimastruct}
\end{figure*}
In this section, we provide a detailed description of the proposed \ac{pima} protocol.
Each frame is divided into two {\em sub-frames}, namely the {\em \ac{pia} sub-frame} and the {\em \ac{dt} sub-frame}. 
The \ac{pia} sub-frame is used to estimate (at the \ac{bs}) the number of currently active users. Based on this information, the \ac{bs} decides the duration (in slots) of the \ac{dt} sub-frame and assigns each user to one slot, for uplink data transmission.

\subsection{Partial Information Acquisition Sub-frame}

The beginning of frame $t$, and thus of the \ac{pia} sub-frame, is triggered by the \ac{rb}, which is transmitted in broadcast by the \ac{bs} to all users. \acp{rb} mark the start of the \ac{pia} sub-frame and contain the acknowledgments of correctly received packets in the previous frame. Moreover, \acp{rb} provides to each user an accurate \ac{csi} of its channel to the \ac{bs}, denoted as $g_n$. In the \ac{pia} sub-frame the \ac{bs} obtains the estimate  $\hat{\nu}(t)$ of the number of active users $\nu(t)$, as described in detail in Section~\ref{sec:numest}.

At the end of the \ac{pia} sub-frame, the \ac{bs}, knowing $\hat{\nu}(t)$, schedules the transmissions for the next sub-frame.
Let $\bm{q}(t) = [q_1(t), \ldots, q_N(t)]$ be the \textit{slot selection vector}, collecting the slot indices assigned to each user; then, the length of the \ac{dt} sub-frame $L_2(t)$ can be derived from $\bm{q}(t)$ as
\begin{equation}
L_2(t) = \max_{\substack{n}} q_n(t).
\end{equation}
To end the \ac{pia} sub-frame and trigger the beginning of the following \ac{dt} sub-frame, the \ac{bs} transmits the \ac{sb}, which contains the slot selection vector $\bm{q}(t)$, encoded as described in Section \ref{overhead}. 

Note that, to maintain synchronization, inactive users could a) wake up and wait for the next downlink \ac{rb} when generating a packet, or b) always wake up when \acp{rb} and \acp{sb} are transmitted (this can be achieved by collecting timing information in the beacons).

\subsection{Data Transmission Sub-frame} 
In the \ac{dt} sub-frame, users transmit their packets, according to the scheduling set by the \ac{bs} in the \acp{sb}. 
Note that packets generated by user $n$ during the \ac{dt} sub-frame are delayed and transmitted in the following frame, to reduce collisions, since the \ac{dt} frame length is derived only based on the number of users active in the \ac{pia} sub-frame.
Indeed, data transmissions of users who do not request resources in the \ac{pia} subframe would introduce uncertainty in the optimization of $\bm{q}_n(t)$, limiting the ability of \ac{pima} to adapt to instantaneous traffic conditions. 

Since all the derivations in the following are related to each frame separately, to simplify notation, we drop the frame index $t$ from all the variables.

An example of the frame structure and packet transmissions of \ac{pima} is reported in Fig.~\ref{fig:pimastruct}.

\subsection{PIMA Beacons Overhead}\label{overhead}

We consider two options for the coding of $\bm{q}$. The first option provides that all the users in the system receive the explicit indication of the allocated slot in the \ac{dt} sub-frame. Consider a codebook of $N$ codewords (each of length $\log_2N$ bits)  representing all possible sorting of user indices, where if a user is in position $i_n$, its assigned slot is $k = i_n \mod L_2 + 1$. With this codebook, the \ac{sb} introduces an overhead of $R_{\rm SB} = (N+1)\log_2N$ bits, since an additional codeword is needed to indicate the length of the \ac{dt} sub-frame, $L_2$.
Note that for a large number of users, this encoding strategy can significantly increase the length of \ac{sb}, deteriorating the performance of \ac{pima}. 

Therefore, we consider a second strategy, in which all users know a list of random sequences (each of $N\log_2 N$ bits) of length $J$, indicating the order in which the users will be served. 
In this case, the \ac{bs} transmit in the \ac{sb} only the index corresponding to the scheduling sequence and the codeword to indicate $L_2$, providing an overhead of $\log_2J + \log_2 N$ bits. 
In the following, we adopt this second strategy to reduce the \ac{sb} overhead.

\section{Estimation of the Number of Active Users}\label{sec:numest}

To obtain an estimate of the number of active users at the \ac{bs}, each active user transmits an \ac{ai} signal immediately after receiving the \ac{rb}.
Note that we neglect here the propagation time between \ac{bs} and the user, which can be easily accommodated by considering a transition (silent) time between the \ac{rb} and \ac{ai} transmissions.  
The set of users transmitting the \ac{ai} signals during the \ac{pia} sub-frame is 
\begin{equation}
\mathcal{N}_{\rm a} = \{k : B_n>0\},
\end{equation}
with $|\mathcal{N}_{\rm a}| = \nu$. We stress that the \ac{bs} does not know the identity of the active users, since the \ac{ai} signals do not contain such information, to make them shorter.
In particular, we assume that each user transmits a single complex symbol $\gamma_{n}$ in the \ac{pia} sub-frame.  Assuming perfect synchronization, the received signal, at the \ac{bs}, is the superposition of all the symbols transmitted by the users, i.e.,
\begin{equation}
    y = \sum_{k \in \mathcal{N}_{\rm a}} h_n\gamma_n + w,
\end{equation}
where $h_n$ is the channel coefficient between user $n$ and the \ac{bs} and $w$ is the \ac{awgn} term with zero mean and variance $\sigma_w^2$. Assuming that perfect \ac{csi} is obtained by users through the \ac{rb} downlink transmission, each user $n$ perfectly inverts the channel, setting $\gamma_{n} = 1/h_n$, therefor the \ac{bs} receives 
\begin{equation}
y = \nu + w.
\end{equation}

The \ac{map} estimate of the number of active users is then
\begin{equation}\label{hnu}
\hat{\nu} = \argmax_{b}\; p_{y|\nu}(y|b)\;p_{\nu}(b),
\end{equation}
where, for the \ac{awgn} channel, the \ac{pdf} of the received signal conditioned to the number of active users is 
\begin{equation}\label{gauss}
p_{y|\nu}(y|b) = \frac{1}{\sqrt{2\pi\sigma_w^2}} e^{-\frac{|y-b|^2}{\sigma^2_w}}.
\end{equation}
This criterion establishes decision regions on the received signal. Define $\delta_b$ as the distance from $b$ (the value obtained without noise) and the region associated with the decision $b+1$. Then, when $y$ falls in the region $[b-\delta_b, b+1-\delta_{b+1}]$, the decision on the number of active devices is $b$. Since the distance between $b$ and $b+1$ is 1, we have that the distance from $b$ (the value obtained without noise) and the region associated with the decision $b+1$ is $1-\delta_{b+1}$. According to the \ac{map} criterion \eqref{hnu}, the optimal regions satisfy the following equation
\begin{equation}\label{thresholds}
p_{b|\nu}(b+\delta_b|b)p_{\nu}(b) = p_{b+1|\nu}(b+1-(1-\delta_b)|b+1)p_{\nu}(b+1),
\end{equation}
since $\delta_{b+1} = 1-\delta_b$. Replacing \eqref{gauss} into \eqref{thresholds}, after some algebraic steps we have
\begin{equation}
\delta_b =  \frac{1}{2} + \frac{\sigma_w^2}{2}\ln \frac{p_{\nu}(b)}{p_{\nu}(b+1)}.
\end{equation}
The error probability is the probability of falling out of the correct decision region, thus, conditioned on the fact that $\nu=b$ users are active, we have 
\begin{equation}
P_e(b) = {\rm Q}\left(\frac{\delta_b}{\sigma_w/\sqrt{2}}\right) + {\rm Q}\left( \frac{1-\delta_{b+1}}{\sigma_w/\sqrt{2}}\right),
\end{equation}
where  ${\rm Q}(\cdot)$ is the tail distribution function of the standard normal distribution. 

\vspace{5pt}\noindent\emph{I.i.d. Activations.} In the case of i.i.d. activations, the activation process coincides with the packet generation process. Assuming that each user generates packets according to a temporal Poisson process with parameter $\lambda_n = \lambda, \forall n$, we obtain
\begin{equation}\label{pN_iid}
    p_{\nu}(b) = \binom{N}{b} b(1-e^{-\lambda T_{\rm a}}) (N-b)e^{-\lambda T_{\rm a}},
\end{equation}
where $T_{\rm a}$ is the time interval considered for the packet generations and $\binom{N}{b}$ counts for all the possibilities of having exactly $b$ active users. Note that for $N\rightarrow \infty$ we have $p_{\nu}(b) \approx p_{\nu}(b+1)$, thus $\delta_b \rightarrow \frac{1}{2}$ and all regions have the same size. In this asymptotic scenario, we also have 
\begin{equation}
P_e(b) \rightarrow 2{\rm Q}\left(\frac{1}{\sqrt{2}\sigma_w}\right).
\end{equation}

\section{Frame Efficiency-based Scheduling}\label{sec:frameeff}

In this section, we propose a time-resource scheduling conditioned on the number of active users estimated in the \ac{pia} sub-frame.

First, we introduce a performance metric that takes into account both the packet latency and the collision probability, whose optimization aims at finding the right balance between the two.
Let $l \in \{1,\ldots, L_2\}$ be the slot index within the frame (in the \ac{dt} sub-frame). We define the success indicator function in slot $l$ as $c_l = 1$ if a successful transmission occurs in slot $l$ and $c_l = 0$ otherwise.
Note that the latter case considers both the collision and non-transmission cases.
Then, the \textit{conditional frame efficiency} is defined as the ratio between the number of successes in the frame and the length of the \ac{dt} sub-frame, i.e., 
\begin{equation}\label{frameeff}
    \eta = \frac{1}{L_2}\sum_{l=1}^{L_2}\mathbb{E}[c_l|\nu].
\end{equation}
The adaptive maximization of this metric provides the proper balance between the \ac{dt} sub-frame length and the successful transmission probability.

At each frame, immediately after the end of the \ac{pia} sub-frame, the \ac{bs} solves the following optimization problem:
\begin{subequations}\label{feopt}
	\label{maxprob}
	\begin{equation}
\max_{\substack{\bm{q}}}\eta,
	\end{equation}
	\begin{equation}
{\rm     s.t.\;}		\;\; q_n \in \{1,\ldots, L_2\}.
	\end{equation}
\end{subequations}

The optimization problem \eqref{feopt} is one of \ac{minlp}, and its solution quickly becomes infeasible with long queues or many users. For these reasons, in the following, we focus on the analysis of the i.i.d. activations scenario, designing the parameters of the \ac{pima} scheduler based on its basic assumptions. 
Note that, while the i.i.d. activations scenario could substantially differ from the correlated activations one at high traffic, it still represents a good approximation in low traffic conditions, wherein retransmissions due to collisions occur sporadically. Moreover, the analysis of i.i.d. activations is useful when there is no information available on the transmission correlation at the \ac{bs} or when obtaining such information is too expensive.

\section{Scheduling With I.I.D. Activations}\label{sec:iidactivations}

First, we observe that, since activations of users are i.i.d., we only have to determine how many users are assigned to each slot, as any specific assignment satisfying this constraint will yield the same collision probabilities, thus the same expected frame efficiency. 
Note that in case of decoding failure of multiple packets, the user scheduling should be randomized to avoid the repetition of the same collisions.

To minimize the number of users assigned to the same slot, given a length $L_2$, we assign to slot $l$ the following number of users
\begin{equation}\label{udef}
	u_{l} = 
	\begin{dcases}
		\left\lceil\frac{N}{L_2}\right\rceil \quad {\rm if} \quad l \leq N\bmod{L_2},\\
		\left\lfloor\frac{N}{L_2}\right\rfloor \quad {\rm if} \quad l > N\bmod{L_2},
	\end{dcases}
\end{equation}
where we possibly schedule one more user in the first $\left\lfloor\frac{N}{L_2}\right\rfloor$ slots to minimize the transmission delay.

With this scheduling policy, the slot success random variable $c_l$ can be rewritten as a function of $u_l$, as it only depends on the number of users scheduled in slot $l$. Thus, the optimization problem \eqref{feopt} is reduced to the optimization of the \ac{dt} sub-frame length, $L_2$, and from \eqref{frameeff}, we have
\begin{subequations}
	\label{iidproblem}
	\begin{equation}
		L_2^* = \argmax_{\substack{L_2}} \frac{1}{L_2}\sum_{l=1}^{L_2}\mathbb{E}[c_{l}|\nu, u_l],
	\end{equation}
	\begin{equation}
		\mbox{s.t. } L_2 \in \mathbb{N}/\{0\}.\label{constraintISA}
	\end{equation}
\end{subequations}

Now, given $\nu$, the probability that user $n$ is the one and only active user assigned to slot $l$ is derived by considering all cases of active users, where user $n$ is active and all other users assigned to slot $l$ are, instead, inactive. 
Consider the matrix $\left\lceil N/L_2\right\rceil \times L_2$, where column $l$ collects the user indexes assigned to slot $l$. For slot $l$, assuming that an active user is assigned to slot $l$, the number of favorable cases is given by all the possibilities to place $\nu-1$ users (the remaining active users) in all columns of the matrix indexed by $l' \neq l$. Excluding column $l$, the available positions (row and column index couples) are $N-u_l$. Therefore, the favorable case is given by all the combinations of $\nu-1$ active users taken from $N-u_l$ users. 
Instead, the total number of cases is given by all the combinations of $\nu$ active users chosen from the $N$ scheduled users.
The probability of having a successful transmission in slot $l$ is therefore
\begin{equation}\label{Ec_l}
\mathbb{E}[c_{l}|\nu, u_l] = u_l\frac{\binom{N-u_l}{\nu-1}}{\binom{N}{\nu}},
\end{equation}
where factor $u_l$ counts the users assigned to slot $l$. Note that, in \eqref{Ec_l}, the numerator counts the number of combinations giving exactly one active user assigned to slot $l$, while the denominator counts the total number of possible combinations of active users. The probability of collision in slot $l$ is therefore $1 - \mathbb{E}[c_{l}|\nu, u_l]$. 

The problem \eqref{iidproblem} is a \ac{minlp} problem, and its solution strictly depends on the number of users in the system. If the number of active users is comparable with $N$, \eqref{iidproblem} is not solvable by continuous relaxation of $L_2$, as the rounding functions are not differentiable. However, it is possible to find the optimal frame length $L_2^{*}$ with complexity $O(\log N)$, using a binary search algorithm, or alternatively using a discrete gradient ascent algorithm. In any case,   $L_2^*$ depends only on $\nu$, thus can be computed offline and then stored in a table.
Instead, if $N\rightarrow \infty$ and $\nu << N$ is finite, the following result holds:

\begin{theorem}\label{theorem1}
For $N\rightarrow\infty$, under i.i.d. activations, given a finite number of active users $\nu$, the \ac{dt} sub-frame length $L_2$ maximizing frame efficiency is exactly $\nu$.
\end{theorem}
\begin{proof}
Since $N\rightarrow\infty$ and $1<\nu<<N$, $u_l = \frac{N}{L_2}$ for all $l = 1, \ldots, L_2$. From \eqref{Ec_l}  we have
\begin{equation}
\begin{split}
    \eta &= \mathbb{E}\left[c_{l}\bigg|\nu, \frac{N}{L_2}\right] = \frac{N}{L_2}\binom{N-\frac{N}{L_2}}{\nu-1}\bigg/\binom{N}{\nu} \\&= \frac{\nu}{L_2}\frac{(N-\frac{N}{L_2})!(N-\nu)!}{(N-\frac{N}{L_2}-\nu -1)!(N-1)!},
\end{split}
\end{equation}
then, using the Stirling factorial approximation $\alpha! = \sqrt{2\pi \alpha}(\frac{\alpha}{e})^\alpha$, whose validity is verified with good accuracy even for small values of $\alpha$, with some algebraic steps we obtain
\begin{equation}
\begin{split}
    \eta = &\frac{\nu}{L_2}\overbrace{\sqrt{\frac{N(1-\frac{1}{L_2})(N-\nu)}{[N(1-\frac{1}{L_2}) - \nu +1](N-1)}}}^{\mu(N, \nu, L_2)}  \times \\ \times & \underbrace{\frac{[N(1-\frac{1}{L_2})]^{N(1-\frac{1}{L_2})}(N-\nu)^{N-\nu}}{[N(1-\frac{1}{L_2}) - \nu +1]^{N(1-\frac{1}{L_2}) - \nu +1}(N-1)^{N-1}}}_ {\xi(N, \nu, L_2)}.
\end{split}
\end{equation}
Taking the limit for $N\rightarrow\infty$ we have
\begin{subequations}
\begin{alignat}{2}
        \lim_{N\rightarrow \infty} \mu(N, \nu, L_2) &= 1,\\
        \lim_{N\rightarrow \infty} \xi(N, \nu, L_2) &= \left(1-\frac{1}{L_2}\right)^{\nu -1}.
\end{alignat}
\end{subequations}
Therefore the maximum frame efficiency only depends on $\nu$ and $L_2$, and it is given by
\begin{equation}
    \eta = \frac{\nu}{L_2}\left(1-\frac{1}{L_2}\right)^{\nu -1}.
\end{equation}
Its stationary points are derived from the first order derivative with respect to $L_2$ as
\begin{equation}
\frac{d}{dL_2}\left[\frac{\nu}{L_2}\left(1-\frac{1}{L_2}\right)^{\nu - 1}\right] = 0 \Leftrightarrow L_2 = \nu \vee L_2 = 1,
\end{equation}
thus, while $L_2=1$ is trivially the global minimum for $\nu>1$, the frame efficiency is maximized for $L_2 = \nu$. 
\end{proof} 

\subsection{On the Optimality of the Single Slot Allocation}

Throughout the paper, we have assumed that each user $n$ is assigned to a single slot $q_n$ in the frame. However, we may wonder if this scheduling policy is optimal or if it is preferable to assign multiple slots to each user. Focusing on the case of i.i.d. activations, we have the following result.

\begin{theorem}\label{theorem2}
Under i.i.d. activations, the assignment of a single slot to each user in the \ac{dt} sub-frame is optimal, i.e., it maximizes the expected frame efficiency. 
\end{theorem}
\begin{proof}
Since the collision probability is the same for all users and depends only on the number of other users transmitting in the same slot, by allocating more slots to each user, we increase the collision probability. Hence, single-slot scheduling is optimal in this case.
\end{proof}

\section{Numerical Results}\label{sec:numericalresults}

In this section, we present the numerical results, comparing our \ac{pima} protocol with the state-of-the-art \ac{oma} schedulers in a) i.i.d. activation scenario, b) correlated activation scenario, both with $N = 50$ users, and c) bursty activation scenario, with a large $N$.
Following the assumption of \textit{short packet transmission}, the number of symbols in each packet (slot duration) is  $T_{\rm s} = 10$ \ac{usd}. 
\footnote{The \ac{usd} is the inverse of the bandwidth if the Nyquist
sampling rate is used. All the lengths of the sequences, slots, and beacons are given in \ac{usd}.} Furthermore, a constant \ac{snr} of $10$~dB is assumed at the receiver. 
Note that while the analysis presented in Section~\ref{sec:frameeff} assumes a perfect estimation of the number of active users, the results shown below are obtained using the estimated $\hat{\nu}$. 

For performance comparison, we consider a) the standard \ac{tdma}, which provides fixed-duration frames of $N$ slots, with one user assigned per slot deterministically, b) a stabilized version of the \ac{saloha} protocol, and c) the \ac{cra}-2 protocol of \cite{Choi20On} with preambles of length $M_{\rm p} = N/2$ and $N$. 

\vspace{5pt}\noindent\emph{Stabilized \Acl{saloha}.} For the \ac{saloha} protocol, we consider Rivest's stabilized \ac{saloha} \cite[Chapter~4]{bertsekas21}, \cite{Rivest87Network}, where all users generating packets at slot $l$ are backlogged with the same probability of backlog. The backoff probability is computed at each user through a pseudo-Bayesian algorithm based on an estimate of the number of backlogged nodes $G(l)$ as 
\begin{equation}
    \alpha(l) = \min\left(1, \frac{1}{G(l)}\right),
 \end{equation}
where 
\begin{equation}
    G(l) = \begin{cases}
        G(l-1) + N\theta + (e-2)^{-1}\quad \text{if } c_l = 0,\\
        \max(N\theta, G(l-1)+ N\theta-1) \quad \text{if } c_l = 1,
    \end{cases}
\end{equation}
is the estimated number of backlogged users (with $G(0) = 0$) and $\theta = 1-e^{-\lambda}$ is the packet generation of probability at slot $l$.

\vspace{5pt}\noindent\emph{Modified \ac{cra}-2 Protocol.}
The \ac{cra}-2 protocol was proposed in \cite{Choi20On}. Similarly to \ac{pima}, each frame includes two sub-frames, and in the first sub-frame active users are identified with a temporary identifier, rather than just counted. To this end, each active user randomly chooses and transmits a sequence of complex symbols of length $M_{\rm p}$ ({\em preamble}), from a preamble pool known to both users and \ac{bs}. The \ac{bs} receives all preambles simultaneously and detects them.
Preambles are used as temporary identifiers for the active users, and the \ac{bs} schedules the data transmission in the second sub-frame by allocating one slot per each detected preamble. Here, we assume that preambles are orthogonal; therefore, the number of preambles equals their length $M_{\rm p}$. When $M_{\rm p} = N$, each user is uniquely assigned to a preamble, while for $M_{\rm p} = N/2$, active users choose their preamble for a pool of $M_{\rm p}$ orthogonal preambles uniformly at random. In the latter case, if two or more users transmit the same preamble (and the \ac{bs} detect it), they are assigned the same slot and collide.
For both schemes, we consider the probability of misdetection in the preamble in the first subframe $P_{\rm md}=0.1$. 
Note that in \cite{Choi20On} preambles are assumed to be non-orthogonal, and a \ac{cs}-based algorithm is applied. However, preamble detection in the presence of noise is not considered and a fixed misdetection probability is assumed. The impact on system performance of different \ac{cs} algorithms has been discussed for the case of multiple measurements (e.g., multiple antennas in \ac{bs}) in \cite{Choi18Stability}. Although on one hand \ac{cs}-based detection allows increasing the number of preambles and reducing the probability of collision, on the other hand, it also produces a high probability of misdetection when \ac{bs} is equipped with a single antenna and a large number of users are active \cite{Semper18Sparsity}.

Comparing \ac{pima} with \ac{cra}-2, we have two main differences: a) the first sub-frame is shorter for \ac{pima} than for \ac{cra}-2, and b) the second sub-frame has the same length for both approaches. Thus, on the one hand, overall the \ac{pima} frame is shorter, reducing the average number of packets generated in each frame, and this reduces the average number of packets accumulated in the buffers before transmission in the next frame. On the other hand, in \ac{pima} users may collide in the second sub-frame due to a non-orthogonal allocation, which increases (with respect to \ac{cra}-2) the number of packets to be re-transmitted in the next frame (thus increasing the average number of users in buffers). Lastly, a wrong counting of users in the \ac{pia} sub-frame or a wrong identification in the first \ac{cra}-2 sub-frame increases the collision probability of both schemes. The numerical results presented in this section will compare the performance of both schemes, taking into account the various described effects.

\vspace{5pt}\noindent\emph{Overhead Comparison.}
Since the acknowledgment overhead is neglected, both \ac{tdma} and \ac{saloha} do not entail any overhead. 
For both \ac{pima} and the preamble-based solutions, we assume that a 64-QAM modulation is used to modulate the \ac{sb}, and the list of the scheduling sequence is set to $J = 64$. Under the aforementioned assumptions, the total overhead induced by the \ac{pia} sub-frame is constant and equal to $L_1 = 3$ QAM symbols.
Instead, for preamble-based approaches, the overhead (in symbols) is given by the preamble with length $M_{\rm p}$ and the \ac{bs} feedback, which is typically longer than the \ac{sb} of \ac{pima}. In particular, assuming that $\nu$ preambles are detected by the \ac{bs}, the identifiers of these preambles are fed back in broadcast, in the order of slot allocation (for the subsequent data transmission), providing a total overhead of $M_{\rm p} + \nu$. 

The overhead of \ac{rb} is neglected and does not play any role in the performance comparison, as its length is comparable for all the considered schedulers.

\vspace{5pt}\noindent\emph{Performance metrics.}
For the i.i.d. and correlated activation case, performance is assessed in terms of both \textit{average frame efficiency} $\bar{\eta}$ and \textit{average latency} $\bar{D}$. The former metric is the average frame efficiency conditioned on $\hat{\nu}>0$, obtained by averaging over all the frames with at least one estimated active user. Instead, in the latter metric the average is computed among all successfully delivered packets of user $n$.
Moreover, in the case of i.i.d. activations, we also consider the packet dropping probability $P_{\rm drop}$, counting the packets dropped due to both collisions and replacements in the unit-length buffers when generations occur. 
This probability is constantly $0$, for all the compared schemes, in the correlated activation case, due to the possible retransmission and the infinite-length queues.

Finally, for bursty activations, the performance of the system is evaluated in terms of \textit{burst transmission time} $D_{\rm B}$, i.e., the time needed to transmit all the packets generated in a traffic burst.

\subsection{I.I.D. Activations}
We first report and discuss the results obtained for i.i.d. user activity and $N = 50$ users. In this activation scenario, retransmissions are not allowed. Therefore,  \ac{saloha}  does not include backlogging, and each user attempts the transmission immediately upon the packet generation, i.e., $\alpha(l) = 1, \forall l$.

\begin{figure}
    \centering
    \setlength\fwidth{0.82\columnwidth}
    \setlength\fheight{0.61\columnwidth}
    \definecolor{TDMAcolor}{rgb}{0.49400,0.18400,0.55600}%
\definecolor{SAL0HAcolor}{rgb}{0.46600,0.67400,0.18800}%
\definecolor{PIMAcolor}{rgb}{0.85000,0.32500,0.09800}%
\definecolor{CHOIcolor}{rgb}{0.00000,0.44700,0.74100}%
\pgfplotsset{every tick label/.append style={font=\footnotesize}}

\begin{tikzpicture}

\begin{axis}[%
    width=\fwidth,
    height=\fheight,
    at={(0\fwidth,0\fheight)},
    scale only axis,
    ylabel style={font=\normalsize},
    xlabel style={font=\normalsize},
    xmin=0.01,
    xmax=5,
    xlabel style={font=\color{white!15!black}},
    xlabel={Packet Generation Rate $\Lambda$ [pkt/slot]},
    ymin=0,
    ymax=0.6,
    yminorticks=true,
    ylabel style={font=\color{white!15!black}},
    ylabel={Avg. Frame Efficiency $\bar{\eta}$},
    axis background/.style={fill=white},
    xmajorgrids,
    ymajorgrids,
    legend style={at={(0.28,0.02)}, legend columns = 2, anchor=south west, legend cell align=left, align=left, font=\footnotesize, draw=white!15!black}
]

\addplot [color=TDMAcolor, very thick, mark size=2.0pt, mark=o, mark options={solid, TDMAcolor}]
  table[row sep=crcr]{%
0.01	0.0206083650190114\\
0.564444444444445	0.0586392845313029\\
1.11888888888889	0.106739512685256\\
1.67333333333333	0.15434448265083\\
2.22777777777778	0.200042500000001\\
2.78222222222222	0.243445\\
3.33666666666667	0.284482499999999\\
3.89111111111111	0.322852499999998\\
4.44555555555556	0.35922125\\
5	0.393695000000005\\
};
\addlegendentry{TDMA}

\addplot [color=CHOIcolor, dashed, very thick, mark size=2.0pt, mark=x, mark options={solid, CHOIcolor}]
  table[row sep=crcr]{%
0.01	0.270562478647662\\
0.564444444444445	0.280778458113135\\
1.11888888888889	0.291943602354874\\
1.67333333333333	0.305728718806296\\
2.22777777777778	0.320984606537988\\
2.78222222222222	0.338413798428694\\
3.33666666666667	0.357510219401353\\
3.89111111111111	0.378231478799828\\
4.44555555555556	0.400762928566462\\
5	0.424773212974161\\
};
\addlegendentry{CRA-2, $M_{\rm p} = N/2$}

\addplot [color=PIMAcolor, very thick, mark size=2.0pt, mark=star, mark options={solid, PIMAcolor}]
  table[row sep=crcr]{%
0.01	0.00357530019829897\\
0.564444444444445	0.166201494300838\\
1.11888888888889	0.279140857868302\\
1.67333333333333	0.360045106560001\\
2.22777777777778	0.420057033744072\\
2.78222222222222	0.465449241465\\
3.33666666666667	0.499443020741062\\
3.89111111111111	0.52519672753934\\
4.44555555555556	0.542905415037427\\
5	0.554455609248852\\
};
\addlegendentry{PIMA}

\addplot [color=CHOIcolor, very thick, mark size=2.0pt, mark=x, mark options={solid, CHOIcolor}]
  table[row sep=crcr]{%
0.01	0.16219540113662\\
0.564444444444445	0.182558339303243\\
1.11888888888889	0.205901887549965\\
1.67333333333333	0.2325718823371\\
2.22777777777778	0.263303316578361\\
2.78222222222222	0.298603131501192\\
3.33666666666667	0.336689131004079\\
3.89111111111111	0.378660339774514\\
4.44555555555556	0.423151521531646\\
5	0.468218213453245\\
};
\addlegendentry{CRA-2, $M_{\rm p} = N$}

\end{axis}
\end{tikzpicture}%
    \caption{Average frame efficiency versus the total packet generation rate for $N = 50$ and i.i.d. activations.}
    \label{fig:frameeffiid}
\end{figure}
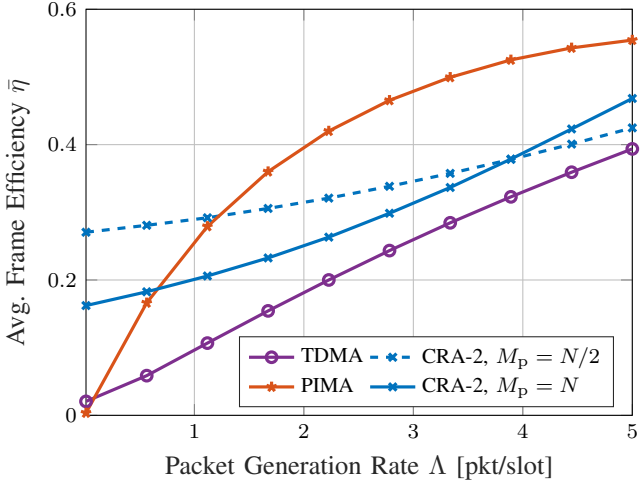

\begin{figure}
    \centering
    \setlength\fwidth{0.82\columnwidth}
    \setlength\fheight{0.61\columnwidth}
    \definecolor{TDMAcolor}{rgb}{0.49400,0.18400,0.55600}%
\definecolor{SAL0HAcolor}{rgb}{0.46600,0.67400,0.18800}%
\definecolor{PIMAcolor}{rgb}{0.85000,0.32500,0.09800}%
\definecolor{CHOIcolor}{rgb}{0.00000,0.44700,0.74100}%
\pgfplotsset{every tick label/.append style={font=\footnotesize}}

\begin{tikzpicture}

\begin{axis}[%
    width=\fwidth,
    height=\fheight,
    at={(0\fwidth,0\fheight)},
    scale only axis,
    ylabel style={font=\normalsize},
    xlabel style={font=\normalsize},
    xmin=0.01,
    xmax=5,
    xlabel style={font=\color{white!15!black}},
    xlabel={Packet Generation Rate $\Lambda$ [pkt/slot]},
    ymin=0,
    ymax=255,
    yminorticks=true,
    ylabel style={font=\color{white!15!black}},
    ylabel={Avg. Packet Latency $\bar{D}$ [usd]},
    axis background/.style={fill=white},
    xmajorgrids,
    ymajorgrids,
    legend style={at={(0.02,0.62)}, legend columns=2, anchor=south west, legend cell align=left, align=left, font=\footnotesize, draw=white!15!black}
]

\addplot [color=TDMAcolor, very thick, mark size=2.0pt, mark=o, mark options={solid, TDMAcolor}]
  table[row sep=crcr]{%
0.01	251.227985647908\\
0.564444444444445	248.0604371186\\
1.11888888888889	244.94371329855\\
1.67333333333333	243.018530361126\\
2.22777777777778	240.982807463035\\
2.78222222222222	238.773258599067\\
3.33666666666667	236.334930734437\\
3.89111111111111	233.577627111954\\
4.44555555555556	231.145956038172\\
5	228.975229260849\\
};
\addlegendentry{TDMA}

\addplot [color=CHOIcolor, dashed, very thick, mark size=2.0pt, mark=x, mark options={solid, CHOIcolor}]
  table[row sep=crcr]{%
0.01	53.2969568609287\\
0.564444444444445	55.1587204491713\\
1.11888888888889	57.6899329064134\\
1.67333333333333	60.4777195149448\\
2.22777777777778	63.5442946701488\\
2.78222222222222	67.0742123089087\\
3.33666666666667	70.9249204381788\\
3.89111111111111	75.382608290075\\
4.44555555555556	80.2073178961149\\
5	85.671460361471\\
};
\addlegendentry{CRA-2, $M_{\rm p} = N/2$}

\addplot [color=SAL0HAcolor, very thick, mark size=2.0pt, mark=square, mark options={solid, SAL0HAcolor}]
  table[row sep=crcr]{%
0.01	4.87368249279553\\
0.564444444444445	4.99102293511495\\
1.11888888888889	5.00076576417113\\
1.67333333333333	5.00221976781867\\
2.22777777777778	5.00861351326839\\
2.78222222222222	4.99900071992082\\
3.33666666666667	4.99388377632859\\
3.89111111111111	4.99012857914115\\
4.44555555555556	5.00054384528813\\
5	4.99072607396152\\
};
\addlegendentry{SALOHA}

\addplot [color=CHOIcolor, very thick, mark size=2.0pt, mark=x, mark options={solid, CHOIcolor}]
  table[row sep=crcr]{%
0.01	92.1535408026256\\
0.564444444444445	97.5432600627752\\
1.11888888888889	102.066517101103\\
1.67333333333333	107.302422234455\\
2.22777777777778	113.32352325406\\
2.78222222222222	120.280760719903\\
3.33666666666667	128.05134242275\\
3.89111111111111	136.811388435497\\
4.44555555555556	147.50110997601\\
5	159.481971480041\\
};
\addlegendentry{CRA-2, $M_{\rm p} = N$}

\addplot [color=PIMAcolor, very thick, mark size=2.0pt, mark=star, mark options={solid, PIMAcolor}]
  table[row sep=crcr]{%
0.01	16.0091918145361\\
0.564444444444445	16.3687679858046\\
1.11888888888889	16.7170637993806\\
1.67333333333333	17.1390638910017\\
2.22777777777778	17.6268605883218\\
2.78222222222222	18.1510769793659\\
3.33666666666667	18.794148588801\\
3.89111111111111	19.5083008887317\\
4.44555555555556	20.4069387486864\\
5	21.4677739032573\\
};
\addlegendentry{PIMA}

\end{axis}
\end{tikzpicture}%
    \caption{Average latency versus the total packet generation rate for $N = 50$ and i.i.d. activations.}
    \label{fig:latencyiid}
\end{figure}
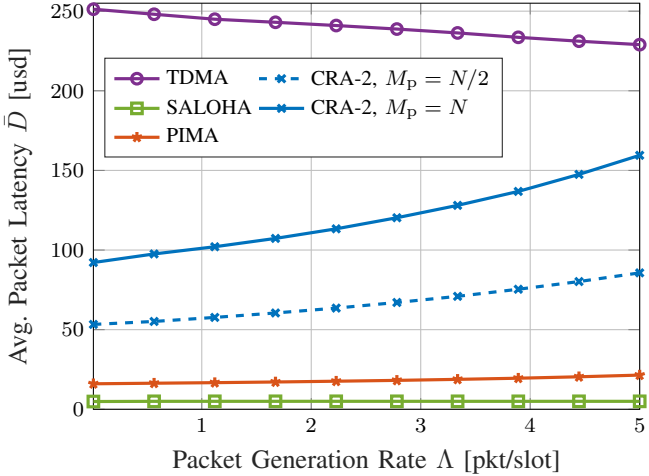

A comparison of the average frame efficiency achieved by each of the schemes is shown in Fig.~\ref{fig:frameeffiid}, as a function of the total packet generation rate $\Lambda$. While this metric cannot be defined for \ac{saloha}, as it does not divide time into frames, we observe that \ac{tdma}, adopting the constant frame length, provides a very low frame efficiency.  
Moreover, since only the frames wherein $\hat{\nu}>0$ are counted on average, \ac{pima} achieves a low frame efficiency at extremely low traffic due to the noise affecting the estimation of $\nu$. Indeed, with sporadic activity, very few frames see active users, and, in many of these rare cases, we have $\hat{\nu}>0$ due to the noise contribution.
As the traffic intensity increases, instead, \ac{pima} achieves the highest frame efficiency, outperforming both the \ac{tdma} and the preamble-based schemes, whose overhead severely affects performance.

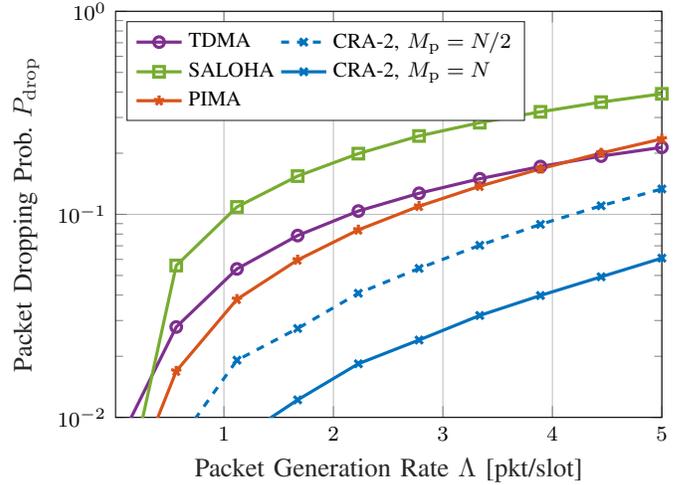
\begin{figure}
    \centering
    \setlength\fwidth{0.82\columnwidth}
    \setlength\fheight{0.61\columnwidth}
    \definecolor{TDMAcolor}{rgb}{0.49400,0.18400,0.55600}%
\definecolor{SAL0HAcolor}{rgb}{0.46600,0.67400,0.18800}%
\definecolor{PIMAcolor}{rgb}{0.85000,0.32500,0.09800}%
\definecolor{CHOIcolor}{rgb}{0.00000,0.44700,0.74100}%
\pgfplotsset{every tick label/.append style={font=\footnotesize}}

\begin{tikzpicture}

\begin{axis}[%
    width=\fwidth,
    height=\fheight,
    at={(0\fwidth,0\fheight)},
    scale only axis,
    ylabel style={font=\normalsize},
    xlabel style={font=\normalsize},
    xmin=0.01,
    xmax=5,
    ymode = log,
    xlabel style={font=\color{white!15!black}},
    xlabel={Packet Generation Rate  $\Lambda$ [pkt/slot]},
    ymin=0.01,
    ymax=1,
    yminorticks=true,
    ylabel style={font=\color{white!15!black}},
    ylabel={Packet Dropping Prob. $P_{\rm drop}$},
    axis background/.style={fill=white},
    xmajorgrids,
    ymajorgrids,
    legend style={at={(0.02,0.73)}, anchor=south west, legend columns = 2, legend cell align=left, align=left, font=\footnotesize, draw=white!15!black}
]

\addplot [color=TDMAcolor, very thick, mark size=2.0pt, mark=o, mark options={solid, TDMAcolor}]
  table[row sep=crcr]{%
0.01	0.007\\
0.564444444444445	0.0278735695301689\\
1.11888888888889	0.0538608835055219\\
1.67333333333333	0.078635834353447\\
2.22777777777778	0.103807449137878\\
2.78222222222222	0.126991384488493\\
3.33666666666667	0.149748383264405\\
3.89111111111111	0.171976686778274\\
4.44555555555556	0.193540511415935\\
5	0.213816791225469\\
};
\addlegendentry{TDMA}

\addplot [color=CHOIcolor, dashed, very thick, mark size=2.0pt, mark=x, mark options={solid, CHOIcolor}]
  table[row sep=crcr]{%
0.01	0\\
0.564444444444445	0.00728736325112147\\
1.11888888888889	0.0191302942707866\\
1.67333333333333	0.0274087228727087\\
2.22777777777778	0.0408569712461218\\
2.78222222222222	0.0541776046161457\\
3.33666666666667	0.0703950259117256\\
3.89111111111111	0.0892578603380042\\
4.44555555555556	0.110210320564702\\
5	0.133421001012782\\
};
\addlegendentry{CRA-2, $M_{\rm p} = N/2$}

\addplot [color=SAL0HAcolor, very thick, mark size=2.0pt, mark=square, mark options={solid, SAL0HAcolor}]
  table[row sep=crcr]{%
0.01	0.002460024600246\\
0.564444444444445	0.0559667889384321\\
1.11888888888889	0.108764625335367\\
1.67333333333333	0.154320388711832\\
2.22777777777778	0.199078248111417\\
2.78222222222222	0.243237064695573\\
3.33666666666667	0.282793377067798\\
3.89111111111111	0.319867662192358\\
4.44555555555556	0.357081121126065\\
5	0.392319059077249\\
};
\addlegendentry{SALOHA}

\addplot [color=CHOIcolor, very thick, mark size=2.0pt, mark=x, mark options={solid, CHOIcolor}]
  table[row sep=crcr]{%
0.01	0\\
0.564444444444445	0.00377158034528552\\
1.11888888888889	0.00772553729229777\\
1.67333333333333	0.0121958542502342\\
2.22777777777778	0.0183726324797874\\
2.78222222222222	0.0240442630568453\\
3.33666666666667	0.0317641084504968\\
3.89111111111111	0.0398313341203033\\
4.44555555555556	0.0492482017310313\\
5	0.0608632493037665\\
};
\addlegendentry{CRA-2, $M_{\rm p} = N$}

\addplot [color=PIMAcolor, very thick, mark size=2.0pt, mark=star, mark options={solid, PIMAcolor}]
  table[row sep=crcr]{%
0.01	0.003\\
0.564444444444445	0.0168944537518737\\
1.11888888888889	0.0381471875362086\\
1.67333333333333	0.0596087374409733\\
2.22777777777778	0.0838230025080616\\
2.78222222222222	0.109650160190939\\
3.33666666666667	0.137592042062699\\
3.89111111111111	0.167532147658583\\
4.44555555555556	0.200331438118078\\
5	0.235141561700777\\
};
\addlegendentry{PIMA}

\end{axis}
\end{tikzpicture}%
    \caption{Average packet dropping probability versus the total packet generation rate for $N = 50$ and i.i.d. activations.}
    \label{fig:pdrop}
\end{figure}

Figs.~\ref{fig:latencyiid} and \ref{fig:pdrop} shows the effect of the packet generation rate on the average latency and packet dropping probability, respectively. In this context, all packets generated during a frame transmission wait, on average, $N/2$ slots in low traffic conditions, therefore \ac{tdma} shows the highest latency. Still, latency decreases as the traffic increases, since the buffering delay is reduced by the new packets replacing the older ones in the queue. However, the dropping probability increases up to over $0.1$ in high-traffic conditions.
Instead, the \ac{saloha} scheduler attains the lowest latency in this scenario, transmitting all packets immediately upon generation. It then provides a lower bound on the latency, as all colliding packets are discarded and do not contribute to the evaluation. However, such dropped packets increase the dropping probability, which approaches $1$ for large values of $\Lambda$.
The \ac{cra}-2 scheduler with $N$ preambles, instead, is collision-free, and it drops a reduced number of buffered packets due to its shorter \ac{dt} sub-frame with respect to \ac{tdma}. Indeed, \ac{cra}-2 achieves the lowest $P_{\rm drop}$ among the considered approaches: this improvement comes at the cost of higher latency, due to the longer time needed for the first sub-frame.  
Lastly, while the already mentioned schemes drop packets due either to collision (\ac{saloha}) or new packet generations (\ac{tdma}, \ac{cra}-2 with $M_{\rm p}=N$), \ac{pima} and \ac{cra}-2 with $M_{\rm p}=N/2$ attempt to merge the advantages of the aforementioned solutions.
On one hand, \ac{pima} provides a higher collision probability than the preamble collision probability of \ac{cra}-2. On the other hand, collisions are compensated by a reduced dropping probability of newly generated packets, thanks to a shorter first sub-frame. Thus, while \ac{cra}-2 achieves a considerably lower dropping probability than \ac{pima}, its latency is higher than that of \ac{pima}, which guarantees close-to-minimum latency.
 
\subsection{Correlated Activations}
We now assess the performance of the general case of correlated user activations, assuming infinite buffer lengths and an infinite number of (re)transmission attempts for each packet. We consider a total number of $N = 50$ users, under a traffic intensity guaranteeing queues' stability, as obtained by simulations. Performance results are shown as a function of the traffic intensity  $0.01\leq\Lambda\leq3.5$, with $\Lambda =\lambda_T T_{\rm s}$.

\begin{figure}
    \centering
    \setlength\fwidth{0.82\columnwidth}
    \setlength\fheight{0.61\columnwidth}
    \definecolor{TDMAcolor}{rgb}{0.49400,0.18400,0.55600}%
\definecolor{SAL0HAcolor}{rgb}{0.46600,0.67400,0.18800}%
\definecolor{PIMAcolor}{rgb}{0.85000,0.32500,0.09800}%
\definecolor{CHOIcolor}{rgb}{0.00000,0.44700,0.74100}%
\pgfplotsset{every tick label/.append style={font=\footnotesize}}

\begin{tikzpicture}

\begin{axis}[%
    width=\fwidth,
    height=\fheight,
    at={(0\fwidth,0\fheight)},
    scale only axis,
    ylabel style={font=\normalsize},
    xlabel style={font=\normalsize},
    xmin=0.01,
    xmax=3.5,
    xlabel style={font=\color{white!15!black}},
    xlabel={Packet Generation Rate  $\Lambda$ [pkt/slot]},
    ymin=0,
    ymax=0.5,
    yminorticks=true,
    ylabel style={font=\color{white!15!black}},
    ylabel={Avg. Frame Efficiency $\bar{\eta}$},
    axis background/.style={fill=white},
    xmajorgrids,
    ymajorgrids,
    legend style={at={(0.28,0.02)}, legend columns = 2, anchor=south west, legend cell align=left, align=left, font=\footnotesize, draw=white!15!black}
]

\addplot [color=TDMAcolor, very thick, mark size=2.0pt, mark=o, mark options={solid, TDMAcolor}]
  table[row sep=crcr]{%
0.01	0.0206083650190114\\
0.397777777777778	0.0458858944954137\\
0.785555555555556	0.0800913589645996\\
1.17333333333333	0.118008142812404\\
1.56111111111111	0.156457278579912\\
1.94888888888889	0.194965\\
2.33666666666667	0.2342475\\
2.72444444444444	0.272952499999999\\
3.11222222222222	0.312137499999998\\
3.5	0.35097\\
};
\addlegendentry{TDMA}

\addplot [color=CHOIcolor, dashed, very thick, mark size=2.0pt, mark=x, mark options={solid, CHOIcolor}]
  table[row sep=crcr]{%
0.01	0.270562478647662\\
0.397777777777778	0.277738053261707\\
0.785555555555556	0.285912707689794\\
1.17333333333333	0.294623957918151\\
1.56111111111111	0.305167970915616\\
1.94888888888889	0.316321449957273\\
2.33666666666667	0.329418038016488\\
2.72444444444444	0.343929862774135\\
3.11222222222222	0.360190569234548\\
3.5	0.377952945024947\\
};
\addlegendentry{CRA-2, $M_{\rm p} = N/2$}

\addplot [color=PIMAcolor, very thick, mark size=2.0pt, mark=star, mark options={solid, PIMAcolor}]
  table[row sep=crcr]{%
0.01	0.00360971908766898\\
0.397777777777778	0.124129054854796\\
0.785555555555556	0.217373283010223\\
1.17333333333333	0.290714112293665\\
1.56111111111111	0.348818528050779\\
1.94888888888889	0.394265999050763\\
2.33666666666667	0.430698251728791\\
2.72444444444444	0.457806005705666\\
3.11222222222222	0.475314084553073\\
3.5	0.484138728312115\\
};
\addlegendentry{PIMA}

\addplot [color=CHOIcolor, very thick, mark size=2.0pt, mark=x, mark options={solid, CHOIcolor}]
  table[row sep=crcr]{%
0.01	0.16219540113662\\
0.397777777777778	0.175942957572621\\
0.785555555555556	0.191737628086311\\
1.17333333333333	0.208821646794821\\
1.56111111111111	0.22816334036082\\
1.94888888888889	0.249493731084609\\
2.33666666666667	0.273194220572147\\
2.72444444444444	0.299306850404487\\
3.11222222222222	0.32792047790415\\
3.5	0.358580450142592\\
};
\addlegendentry{CRA-2, $M_{\rm p} = N$}

\end{axis}
\end{tikzpicture}%
    \caption{Average frame efficiency versus the total packet generation rate for $N = 50$ and correlated activations.}
    \label{fig:corrframeeff}
\end{figure}
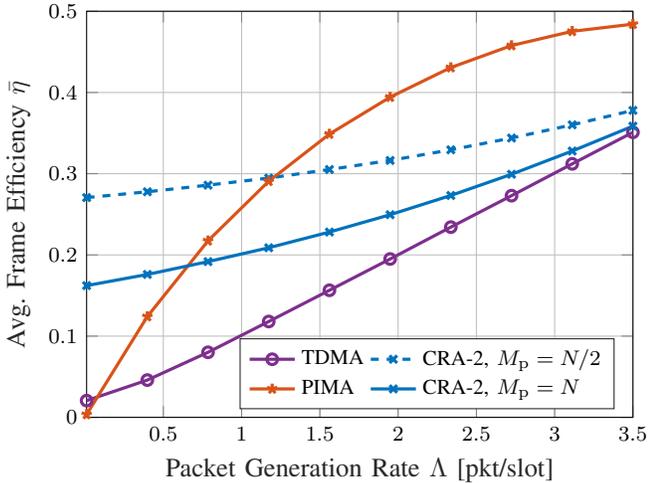

Firstly, Fig.~\ref{fig:corrframeeff} shows the average frame efficiency $\bar{\eta}$, as a function of the total packet generation rate. The performance is almost the same as the i.i.d. case. This is mostly due to the fact that, in stability conditions, few retransmissions are performed multiple times. While all the observations on Fig.~\ref{fig:frameeffiid} still hold, we observe a very slight degradation of \ac{pima}, as it is designed for i.i.d. activation statistics and does not take into account previous collisions.

\begin{figure}
    \centering
    \setlength\fwidth{0.82\columnwidth}
    \setlength\fheight{0.61\columnwidth}
    \definecolor{TDMAcolor}{rgb}{0.49400,0.18400,0.55600}%
\definecolor{SAL0HAcolor}{rgb}{0.46600,0.67400,0.18800}%
\definecolor{PIMAcolor}{rgb}{0.85000,0.32500,0.09800}%
\definecolor{CHOIcolor}{rgb}{0.00000,0.44700,0.74100}%
\pgfplotsset{every tick label/.append style={font=\footnotesize}}

\begin{tikzpicture}

\begin{axis}[%
    width=\fwidth,
    height=\fheight,
    at={(0\fwidth,0\fheight)},
    scale only axis,
    ylabel style={font=\normalsize},
    xlabel style={font=\normalsize},
    xmin=0.01,
    xmax=3.5,
    xlabel style={font=\color{white!15!black}},
    xlabel={Packet Generation Rate $\Lambda$ [pkt/slot]},
    ymin=0,
    ymax=400,
    yminorticks=true,
    ylabel style={font=\color{white!15!black}},
    ylabel={Avg. Packet Latency $\bar{D}$ [usd]},
    axis background/.style={fill=white},
    xmajorgrids,
    ymajorgrids,
    legend style={at={(0.02,0.37)}, legend columns = 2, anchor=south west, legend cell align=left, align=left, font=\footnotesize, draw=white!15!black}
]

\addplot [color=TDMAcolor, very thick, mark size=2.0pt, mark=o, mark options={solid, TDMAcolor}]
  table[row sep=crcr]{%
0.01	251.227985647908\\
0.397777777777778	259.589714713921\\
0.785555555555556	271.322562075595\\
1.17333333333333	282.462414472882\\
1.56111111111111	295.643101300088\\
1.94888888888889	310.348816378483\\
2.33666666666667	326.900231831119\\
2.72444444444444	344.792574330948\\
3.11222222222222	362.926753893595\\
3.5	384.936025333182\\
};
\addlegendentry{TDMA}

\addplot [color=CHOIcolor, dashed, very thick, mark size=2.0pt, mark=x, mark options={solid, CHOIcolor}]
  table[row sep=crcr]{%
0.01	53.2969568609287\\
0.397777777777778	54.8664412384461\\
0.785555555555556	56.8179284277435\\
1.17333333333333	59.1760049180619\\
1.56111111111111	61.7703363440487\\
1.94888888888889	64.7419823473204\\
2.33666666666667	68.2281254803582\\
2.72444444444444	72.4191439617707\\
3.11222222222222	77.3190567552466\\
3.5	83.6727682995812\\
};
\addlegendentry{CRA-2, $M_{\rm p} = N/2$}

\addplot [color=SAL0HAcolor, very thick, mark size=2.0pt, mark=square, mark options={solid, SAL0HAcolor}]
  table[row sep=crcr]{%
0.01	5.00412340928535\\
0.397777777777778	6.43882436715423\\
0.785555555555556	8.22986756202203\\
1.17333333333333	10.9656877074673\\
1.56111111111111	14.9642054013082\\
1.94888888888889	20.6177492642629\\
2.33666666666667	30.1864898208759\\
2.72444444444444	48.177417284734\\
3.11222222222222	91.6901324882264\\
3.5	316.497131675678\\
};
\addlegendentry{SALOHA}

\addplot [color=CHOIcolor, very thick, mark size=2.0pt, mark=x, mark options={solid, CHOIcolor}]
  table[row sep=crcr]{%
0.01	92.1535408026256\\
0.397777777777778	96.2978692162661\\
0.785555555555556	99.9144428683196\\
1.17333333333333	103.562265645203\\
1.56111111111111	107.532320198218\\
1.94888888888889	112.102101920332\\
2.33666666666667	117.494086747127\\
2.72444444444444	123.300886561516\\
3.11222222222222	130.158818262479\\
3.5	137.879750685086\\
};
\addlegendentry{CRA-2, $M_{\rm p} = N$}

\addplot [color=PIMAcolor, very thick, mark size=2.0pt, mark=star, mark options={solid, PIMAcolor}]
  table[row sep=crcr]{%
0.01	15.9438482271639\\
0.397777777777778	16.9423795066083\\
0.785555555555556	18.2608641564928\\
1.17333333333333	20.1327254206771\\
1.56111111111111	22.1554930650563\\
1.94888888888889	25.8438260645377\\
2.33666666666667	31.0934386202192\\
2.72444444444444	40.0914441525833\\
3.11222222222222	59.4397684312152\\
3.5	97.3214336519185\\
};
\addlegendentry{PIMA}

\end{axis}
\end{tikzpicture}%
    \caption{Average packet latency versus the total packet generation rate for $N = 50$ and correlated activations.}
    \label{fig:corrlatency}
\end{figure}
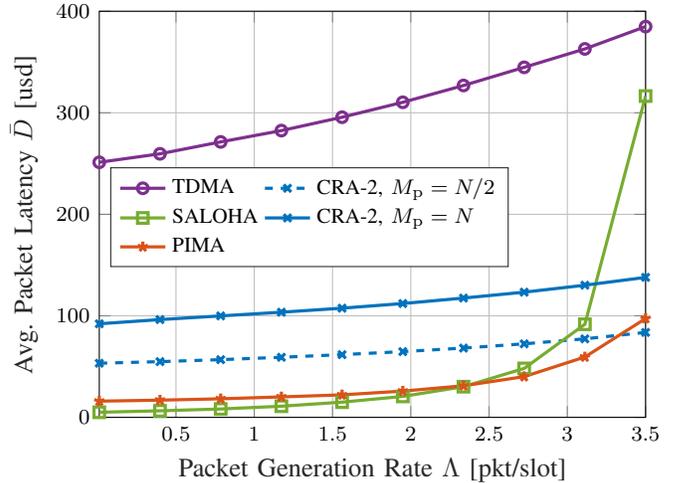

Second, Fig.~\ref{fig:corrlatency} shows the average packet latency as a function of the packet generation rate for $N = 50$ users. Still, the latency of \ac{tdma} is much higher than that of other schemes considered due to its frame length of $N$ (maximum). We also observe that, at low traffic, the \ac{pima} scheduler achieves extremely low latency, comparable to the minimum latency provided by \ac{saloha}. In higher traffic conditions, instead, the \ac{saloha} backlogging mechanism prevents users from transmitting their buffered packets immediately, thus increasing the average latency. This effect is mitigated in \ac{pima}, whose latency is reduced, due to its better ability to adapt to instantaneous traffic load.
For the preamble-based approach, instead, we observe that overhead plays a crucial role in overall latency, and shorter preambles yield better performance, while \ac{pima} is still outperforming \ac{cra}-2.

\subsection{Bursty Activations}
We now investigate the performance in the bursty activations scenario. We first assess the performance obtained for a single burst of intensity $\Lambda_{\rm B}$, and then discuss some constraints on the burst interarrival time $\tau_{\rm B}$.
As discussed in Section~\ref{sec:actstat}, here we assume that a random number of active users $\nu$, out of an arbitrarily large number of users in the system $N$, generate a single short packet to be transmitted in the following time slots. As we consider an arbitrary large $N$, the length of the \ac{dt} sub-frame in \ac{pima} is here derived according to Theorem~\ref{theorem1}. 

The number of packet generations in a burst follows a Poisson distribution $\Lambda_{\rm B}$ in the range $[10, 10000]$, thus, the average number of packet generations is $\Lambda_{\rm B}$.  
For comparison purposes, we consider \ac{cra}-2 with fixed preamble lengths $M_{\rm p} = 1000$ and $10000$. We also consider an ideal solution, wherein the length of the preamble is adapted to the average number of generation of packets, that is, $M_{\rm p} = \lambda_{B}$. Note that this ideal  solution is hardly implementable, as it requires different preamble pools based on the traffic generation rate.
Moreover, we do consider neither \ac{saloha}, as all the packets generated simultaneously collide with probability 1, nor \ac{tdma}, as a large $N$ yields a very long frame length.
Note that a longer list of random sequences (that is, a longer overhead) for the scheduling vector encoding (Section~\ref{overhead}) is necessary as the number of users in the system increases. In the following, we still consider a \ac{pia} sub-frame overhead of $L_1=3$ symbols, which can be easily accommodated by adopting a higher modulation order at the \ac{bs} for \ac{sb} transmission.

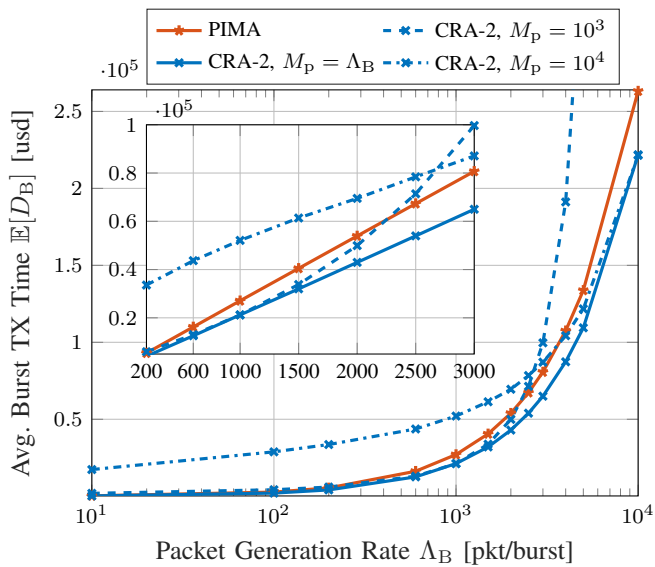
\begin{figure}
    \centering
    \setlength\fwidth{0.82\columnwidth}
    \setlength\fheight{0.61\columnwidth}
    \definecolor{TDMAcolor}{rgb}{0.49400,0.18400,0.55600}%
\definecolor{SAL0HAcolor}{rgb}{0.46600,0.67400,0.18800}%
\definecolor{PIMAcolor}{rgb}{0.85000,0.32500,0.09800}%
\definecolor{CHOIcolor}{rgb}{0.00000,0.44700,0.74100}%
\pgfplotsset{every tick label/.append style={font=\footnotesize}}

\begin{tikzpicture}

\begin{axis}[%
    width=\fwidth,
    height=\fheight,
    at={(0\fwidth,0\fheight)},
    scale only axis,
    ylabel style={font=\normalsize},
    xlabel style={font=\normalsize},
    xmin=10,
    xmax=10000,
    xlabel style={font=\color{white!15!black}},
    xlabel={Packet Generation Rate $\Lambda_{\rm B}$ [pkt/burst]},
    ymin=180,
    ymax=264000,
    xmode = log,
    yminorticks=true,
    ylabel style={font=\color{white!15!black}},
    ylabel={Avg. Burst TX Time $\mathbb{E}[D_{\rm B}]$ [usd]},
    axis background/.style={fill=white},
    xmajorgrids,
    ymajorgrids,
    legend style={at={(0.1,1.2)}, anchor=north west, legend columns = 2, legend cell align=left, align=left, font=\footnotesize, draw=white!15!black}
]

\addplot [color=PIMAcolor, very thick, mark size=2.0pt, mark=star, mark options={solid, PIMAcolor}]
  table[row sep=crcr]{%
10	250.55\\
100	2688.874\\
200	5393.91\\
600	16221.818\\
1000	27070.584\\
1500	40492.44\\
2000	53911.11\\
2500	67318.294\\
3000	80699.424\\
4000	107487.578\\
5000	133754.162\\
10000	263189.63\\
};
\addlegendentry{PIMA}

\addplot [color=CHOIcolor, very thick, dashed, mark size=2.0pt, mark=x, mark options={solid, CHOIcolor}]
  table[row sep=crcr]{%
10	1874.967\\
100	4209.608\\
200	5946.711\\
600	13016.434\\
1000	21206.875\\
1500	33809.353\\
2000	49919.535\\
2500	71399.569\\
3000	99667.545\\
4000	191158.656\\
5000	365266.88\\
10000	6678435.631\\
};
\addlegendentry{CRA-2, $M_{\rm p} = 10^3$}

\addplot [color=CHOIcolor, very thick, mark size=2.0pt, mark=x, mark options={solid, CHOIcolor}]
  table[row sep=crcr]{%
10	187.527\\
100	2022.781\\
200	4090.917\\
600	12571.146\\
1000	21215.342\\
1500	32001.882\\
2000	42968.833\\
2500	53948.615\\
3000	64982.554\\
4000	87325.893\\
5000	109405.028\\
10000	221837.329\\
};
\addlegendentry{CRA-2, $M_{\rm p} =\Lambda_{\rm B}$}


\addplot [color=CHOIcolor, very thick, dashdotted, mark size=2.0pt, mark=x, mark options={solid, CHOIcolor}]
  table[row sep=crcr]{%
10	17319.582\\
100	28828.8\\
200	33579.58\\
600	43685.485\\
1000	52086.622\\
1500	61370.411\\
2000	69496.591\\
2500	78418.539\\
3000	87095.58\\
4000	104293.903\\
5000	121731.424\\
10000	221482.821\\
};
\addlegendentry{CRA-2, $M_{\rm p} = 10^4$}

\end{axis}

\begin{axis}[%
width=0.6\fwidth,
height=0.42\fwidth,
at={(0.1\fwidth,0.26\fwidth)},
scale only axis,
xmin=200,
xmax=3000,
xlabel style={font=\color{white!15!black}},
ymin=5000,
ymax=100000,
xtick={200, 600, 1000, 1500, 2000, 2500, 3000},
xticklabels = {200, 600, 1000, 1500, 2000, 2500, 3000},
axis background/.style={fill=white},
xmajorgrids,
ymajorgrids
]

\addplot [color=PIMAcolor, very thick, mark size=2.0pt, mark=star, mark options={solid, PIMAcolor}]
  table[row sep=crcr]{%
  200	5393.91\\
600	16221.818\\
1000	27070.584\\
1500	40492.44\\
2000	53911.11\\
2500	67318.294\\
3000	80699.424\\
};
\addlegendentry{PIMA}

\addplot [color=CHOIcolor, very thick, dashed, mark size=2.0pt, mark=x, mark options={solid, CHOIcolor}]
  table[row sep=crcr]{%
  200	5946.711\\
600	13016.434\\
1000	21206.875\\
1500	33809.353\\
2000	49919.535\\
2500	71399.569\\
3000	99667.545\\
};
\addlegendentry{CRA-2, $M_{\rm p} = 10^3$}

\addplot [color=CHOIcolor, very thick, mark size=2.0pt, mark=x, mark options={solid, CHOIcolor}]
  table[row sep=crcr]{%
  200	4090.917\\
600	12571.146\\
1000	21215.342\\
1500	32001.882\\
2000	42968.833\\
2500	53948.615\\
3000	64982.554\\
};
\addlegendentry{CRA-2, $M_{\rm p} =\Lambda_{\rm B}$}

\addplot [color=CHOIcolor, very thick, dashdotted, mark size=2.0pt, mark=x, mark options={solid, CHOIcolor}]
  table[row sep=crcr]{%
200	33579.58\\
600	43685.485\\
1000	52086.622\\
1500	61370.411\\
2000	69496.591\\
2500	78418.539\\
3000	87095.58\\
};
\addlegendentry{CRA-2, $M_{\rm p} = 10^4$}

\legend{}
\end{axis}

\end{tikzpicture}%


    \caption{Average burst transmission time versus the packet generations rate (in log scale) in a bursty scenario. The zoomed plot reports results for $600\leq\Lambda_{\rm B}\leq3000$ (in linear scale).}
    \label{fig:burstlatency}
\end{figure}

Fig.~\ref{fig:burstlatency} shows the average burst transmission time as a function of the average number of packet generations. First, we observe a constant gap between \ac{pima} and the ideal preamble-based solution with $M_{\rm p} = \Lambda_{\rm B}$, while fixed-length preambles achieve better performance when the number of active users is comparable to the preamble length. In particular, the \ac{cra}-2 solution with $M_{\rm p} = 1000$ preambles achieves a lower burst transmission time only for $400\leq\Lambda_{\rm B}\leq 2000$, while at least $\Lambda_{\rm B} = 4000$ is needed when $M_{\rm p} = 1000$. 
For a low average number of packet generations, \ac{pima} achieves the best performance due to its reduced overhead. For faster packet generations, both fixed-length preamble approaches suffer from preamble collisions, which implies a much higher packet transmission time due to retransmissions. 
Therefore, while \ac{pima} is able to adapt to all traffic conditions without any change in the \ac{pia} sub-frame, preamble-based approaches should adopt different preamble pools depending on the traffic intensity in order to achieve nice performance.    

\begin{figure}
    \centering
    \setlength\fwidth{0.82\columnwidth}
    \setlength\fheight{0.61\columnwidth}
    \definecolor{PIMAcolor}{rgb}{0.85000,0.32500,0.09800}%
\definecolor{CHOIcolor}{rgb}{0.00000,0.44700,0.74100}%
\pgfplotsset{every tick label/.append style={font=\footnotesize}}

\begin{tikzpicture}

\begin{axis}[%
    width=\fwidth,
    height=\fheight,
    at={(0\fwidth,0\fheight)},
    scale only axis,
    ylabel style={font=\normalsize},
    xlabel style={font=\normalsize},
    xlabel style={font=\color{white!15!black}},
    xlabel={Burst TX Time $D_{\rm B}$ [usd]},
    ymin=1e-2,
    ymax=1,
    ymode = log,
    yminorticks=true,
    ylabel style={font=\color{white!15!black}},
    ylabel={${\rm ECCDF}(D_{\rm B})$},
    axis background/.style={fill=white},
    xmajorgrids,
    ymajorgrids,
    legend style={at={(0.05,1.02)}, legend columns = 2,  anchor=south west, legend cell align=left, align=left, font=\footnotesize, draw=white!15!black}
]

\addplot [color=PIMAcolor, very thick]
  table[row sep=crcr]{%
13498	1\\
14254	0.99\\
14458	0.98\\
14584	0.97\\
14732	0.96\\
14842	0.949\\
14900	0.938\\
14946	0.927\\
15044	0.915\\
15084	0.903\\
15174	0.89\\
15210	0.878\\
15250	0.867\\
15296	0.856\\
15342	0.844\\
15386	0.833\\
15426	0.82\\
15478	0.806\\
15536	0.791\\
15578	0.775\\
15604	0.763\\
15636	0.75\\
15674	0.737\\
15716	0.724\\
15748	0.709\\
15784	0.696\\
15810	0.682\\
15840	0.67\\
15872	0.657\\
15908	0.643\\
15934	0.63\\
15962	0.616\\
15984	0.6\\
16026	0.578\\
16064	0.559\\
16106	0.542\\
16134	0.527\\
16162	0.513\\
16194	0.497\\
16232	0.477\\
16258	0.463\\
16308	0.448\\
16332	0.433\\
16372	0.415999999999999\\
16398	0.399\\
16444	0.383\\
16470	0.367\\
16510	0.353999999999999\\
16554	0.341999999999999\\
16592	0.331\\
16626	0.313\\
16662	0.296999999999999\\
16700	0.280999999999999\\
16736	0.267999999999999\\
16784	0.255999999999999\\
16812	0.243999999999999\\
16844	0.229999999999999\\
16882	0.215\\
16918	0.198\\
16980	0.186\\
17046	0.172\\
17100	0.157\\
17150	0.147\\
17198	0.131\\
17294	0.12\\
17340	0.108\\
17382	0.0939999999999998\\
17482	0.0809999999999997\\
17580	0.0699999999999998\\
17660	0.0589999999999998\\
17744	0.0459999999999999\\
17856	0.0349999999999999\\
17954	0.0249999999999999\\
18216	0.014\\
18636	0.002\\
};
\addlegendentry{PIMA}

\addplot [color=CHOIcolor, very thick]
  table[row sep=crcr]{%
10216	1\\
10953	0.988\\
11217	0.973\\
11360	0.955\\
11448	0.939\\
11575	0.916\\
11657	0.895\\
11756	0.87\\
11822	0.84\\
11888	0.813\\
11949	0.786\\
12004	0.761\\
12070	0.728\\
12130	0.702\\
12191	0.678\\
12246	0.654\\
12301	0.632\\
12361	0.6\\
12416	0.57\\
12477	0.544\\
12543	0.506\\
12598	0.476\\
12659	0.444\\
12713	0.423\\
12774	0.393\\
12840	0.349\\
12912	0.319\\
12972	0.291\\
13033	0.264\\
13099	0.242\\
13165	0.211\\
13215	0.186\\
13269	0.164\\
13330	0.138\\
13407	0.119\\
13473	0.105\\
13566	0.0909999999999999\\
13633	0.0759999999999998\\
13726	0.0629999999999998\\
13803	0.0489999999999998\\
13907	0.0359999999999999\\
14078	0.0259999999999999\\
14309	0.0149999999999999\\
14837	0.005\\
};
\addlegendentry{CRA-2, $M_{\rm p} =\Lambda_{\rm B}$}

\addplot [color=CHOIcolor, very thick, dashed]
  table[row sep=crcr]{%
10765	1\\
11326	0.988\\
11513	0.972\\
11645	0.954\\
11766	0.929\\
11865	0.905\\
11975	0.865\\
12085	0.846\\
12173	0.827\\
12238	0.808\\
12305	0.788\\
12382	0.769\\
12469	0.756\\
12547	0.73\\
12635	0.711\\
12711	0.672\\
12799	0.634\\
12887	0.574\\
12965	0.529\\
13063	0.487\\
13173	0.421\\
13283	0.361\\
13382	0.306\\
13459	0.27\\
13536	0.249\\
13624	0.215\\
13690	0.198\\
13755	0.174\\
13843	0.153\\
13909	0.135\\
13987	0.118\\
14064	0.101\\
14140	0.0789999999999997\\
14272	0.0599999999999998\\
14393	0.0439999999999998\\
14558	0.0319999999999999\\
14734	0.0179999999999999\\
15196	0.00800000000000001\\
};
\addlegendentry{CRA-2, $M_{\rm p} = 10^3$}

\addplot [color=CHOIcolor, very thick, dashdotted]
  table[row sep=crcr]{%
35764	1\\
36160	0.988\\
36336	0.968\\
36457	0.932\\
36567	0.89\\
36677	0.824\\
36787	0.773\\
36897	0.719\\
37007	0.682\\
37128	0.642\\
37260	0.627\\
37458	0.615\\
46105	0.604\\
46259	0.586\\
46391	0.562\\
46501	0.512\\
46611	0.449\\
46721	0.369\\
46831	0.289\\
46941	0.221\\
47051	0.159\\
47161	0.117\\
47282	0.0929999999999997\\
47744	0.0769999999999998\\
56545	0.0599999999999998\\
56721	0.0389999999999999\\
56897	0.0259999999999999\\
57139	0.0149999999999999\\
66556	0.004\\
};
\addlegendentry{CRA-2, $M_{\rm p} = 10^4$}

\addplot [color=PIMAcolor, very thick]
  table[row sep=crcr]{%
75346	1\\
76428	0.99\\
76854	0.98\\
77140	0.969\\
77606	0.959\\
77792	0.948\\
77940	0.937\\
78094	0.927\\
78230	0.915\\
78314	0.905\\
78458	0.894\\
78534	0.884\\
78614	0.873\\
78742	0.862\\
78854	0.85\\
78956	0.839\\
79052	0.826\\
79094	0.815\\
79204	0.804\\
79284	0.792\\
79352	0.78\\
79410	0.768\\
79466	0.758\\
79528	0.747\\
79572	0.735\\
79656	0.722\\
79714	0.71\\
79754	0.697\\
79816	0.687\\
79862	0.675\\
79904	0.663\\
79964	0.652\\
80022	0.641\\
80064	0.627\\
80116	0.613\\
80182	0.602\\
80244	0.588\\
80296	0.578\\
80364	0.567\\
80440	0.557\\
80478	0.542\\
80522	0.529\\
80580	0.516\\
80632	0.506\\
80668	0.494\\
80708	0.48\\
80758	0.469\\
80810	0.458\\
80860	0.444\\
80914	0.433\\
80978	0.42\\
81018	0.407\\
81124	0.391\\
81220	0.378\\
81300	0.365\\
81372	0.355\\
81428	0.343\\
81484	0.332\\
81568	0.321\\
81624	0.309999999999999\\
81674	0.297999999999999\\
81758	0.285999999999999\\
81806	0.272999999999999\\
81874	0.260999999999999\\
81954	0.246999999999999\\
82028	0.235\\
82084	0.222\\
82146	0.209\\
82188	0.197\\
82308	0.185\\
82418	0.171999999999999\\
82536	0.159\\
82632	0.149\\
82712	0.138\\
82778	0.128\\
82840	0.116\\
82962	0.103\\
83124	0.0929999999999997\\
83336	0.0829999999999997\\
83452	0.0729999999999997\\
83614	0.0619999999999998\\
83754	0.0519999999999998\\
83914	0.0419999999999998\\
84118	0.0299999999999999\\
84406	0.0199999999999999\\
84874	0.01\\
85876	0\\
};

\addplot [color=CHOIcolor, very thick]
  table[row sep=crcr]{%
58516	1\\
61142	0.99\\
61406	0.977\\
61593	0.964\\
61802	0.948\\
61989	0.933\\
62143	0.917\\
62297	0.893\\
62429	0.878\\
62572	0.859\\
62693	0.831\\
62825	0.806\\
62946	0.787\\
63056	0.768\\
63166	0.743\\
63276	0.726\\
63397	0.704\\
63526	0.683\\
63639	0.656\\
63738	0.627\\
63848	0.605\\
63980	0.578\\
64087	0.555\\
64211	0.537\\
64373	0.516\\
64538	0.5\\
64662	0.489\\
64758	0.477\\
64838	0.462\\
64959	0.447\\
65121	0.435\\
65286	0.42\\
65407	0.406\\
65583	0.39\\
65704	0.374\\
65836	0.358\\
65968	0.346\\
66100	0.331\\
66243	0.314\\
66386	0.289999999999999\\
66529	0.274999999999999\\
66661	0.250999999999999\\
66782	0.231999999999999\\
66925	0.212999999999999\\
67090	0.199\\
67244	0.175\\
67398	0.159\\
67519	0.143\\
67706	0.122\\
67890	0.105\\
68091	0.0889999999999997\\
68388	0.0789999999999998\\
68729	0.0679999999999998\\
69155	0.0569999999999998\\
69485	0.0459999999999998\\
69826	0.0319999999999999\\
70200	0.0199999999999999\\
70904	0.01\\
74671	0\\
};

\addplot [color=CHOIcolor, very thick, dashed]
  table[row sep=crcr]{%
86019	1\\
91266	0.989\\
92167	0.979\\
93113	0.969\\
93575	0.959\\
93927	0.947\\
94179	0.936\\
94455	0.926\\
94598	0.915\\
94839	0.903\\
94994	0.893\\
95225	0.883\\
95531	0.868\\
95730	0.855\\
95897	0.844\\
96038	0.834\\
96259	0.82\\
96424	0.81\\
96512	0.8\\
96621	0.788\\
96819	0.775\\
96942	0.763\\
97061	0.751\\
97227	0.738\\
97414	0.721\\
97568	0.706\\
97656	0.694\\
97809	0.682\\
97931	0.669\\
98063	0.658\\
98260	0.643\\
98349	0.632\\
98470	0.621\\
98591	0.606\\
98821	0.593\\
98921	0.581\\
99041	0.567\\
99141	0.553\\
99251	0.543\\
99360	0.53\\
99481	0.519\\
99590	0.505\\
99658	0.494\\
99778	0.478\\
99921	0.464\\
100052	0.451\\
100241	0.435\\
100329	0.417\\
100427	0.402\\
100549	0.387\\
100758	0.377\\
100890	0.362\\
101033	0.35\\
101176	0.335\\
101307	0.324\\
101439	0.314\\
101517	0.3\\
101671	0.286999999999999\\
101802	0.274\\
102000	0.260999999999999\\
102132	0.248999999999999\\
102243	0.236999999999999\\
102407	0.225\\
102582	0.213\\
102704	0.202\\
102858	0.189\\
103023	0.179\\
103265	0.168\\
103442	0.157\\
103694	0.146\\
103826	0.135\\
104002	0.124\\
104167	0.11\\
104486	0.0989999999999996\\
104772	0.0879999999999997\\
105036	0.0759999999999997\\
105310	0.0659999999999998\\
105640	0.0559999999999998\\
106037	0.0439999999999998\\
106356	0.0339999999999999\\
107081	0.0239999999999999\\
108193	0.014\\
110326	0.004\\
};

\addplot [color=CHOIcolor, very thick, dashdotted]
  table[row sep=crcr]{%
75112	1\\
75673	0.99\\
75948	0.977\\
76212	0.964\\
76454	0.947\\
76597	0.932\\
76828	0.919\\
76993	0.898\\
77202	0.88\\
77411	0.866\\
77862	0.853\\
79050	0.841\\
85156	0.83\\
85464	0.814\\
85629	0.793\\
85761	0.769\\
85893	0.741\\
86025	0.714\\
86135	0.687\\
86245	0.656\\
86355	0.622\\
86465	0.596\\
86575	0.556\\
86685	0.519\\
86795	0.467\\
86905	0.433\\
87015	0.395\\
87125	0.355\\
87257	0.325\\
87367	0.293999999999999\\
87488	0.264\\
87598	0.234\\
87719	0.217\\
87906	0.199\\
88280	0.187\\
88951	0.176\\
95596	0.165\\
95904	0.155\\
96124	0.143\\
96333	0.129\\
96553	0.113\\
96729	0.0979999999999998\\
96927	0.0849999999999997\\
97147	0.0719999999999998\\
97345	0.0529999999999998\\
97609	0.0399999999999999\\
98071	0.0279999999999999\\
105145	0.0179999999999999\\
106905	0.00800000000000001\\
};
\end{axis}

\draw [dashed] (2,2.1) ellipse (1.5cm and 0.35cm);
\node[] at (2,2.7){\footnotesize $\Lambda_{\rm B} = 600$};
\draw [dashed] (5.3,1.2) ellipse (1.3cm and 0.35cm);
\node[] at (5.3,0.6){\footnotesize $\Lambda_{\rm B} = 3000$};

\end{tikzpicture}%
    \caption{\ac{eccdf} of the burst transmission time.}
    \label{fig:burstccdf}
\end{figure}
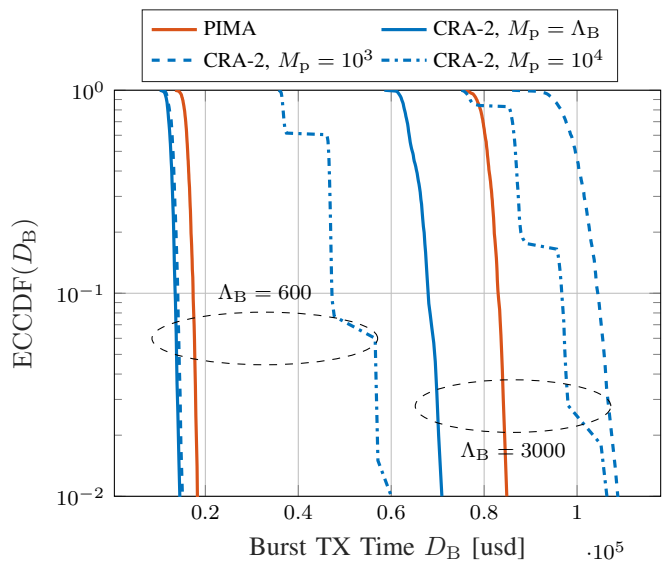

Finally, \ac{eccdf} of the burst transmission time is shown by Fig.~\ref{fig:burstccdf} for $\Lambda_{\rm B} = 600$ and $3000$. The results confirm the comments on Fig.~\ref{fig:burstlatency}, with the ideal case $M_{\rm p} = \Lambda_{\rm B}$ always attaining the lowest transmission time and the fixed-preamble approaches outperforming \ac{pima} only for a preamble length comparable to the number of packet generations.
Note that this last figure gives an indication of the minimum burst interarrival time $\tau_{\rm B}$. To minimize the probability of overlap between two traffic bursts, in particular, $\tau_{\rm B}$ must be large enough to minimize the probability of having new arrivals while there are still previously generated packets to be transmitted. 
Setting a threshold on the overlap probability, the minimum $\tau_{\rm B}$ allowing one to satisfy such threshold can be easily retrieved by \ac{eccdf}.  
For example, from Fig.~\ref{fig:burstccdf}, the probability of overlap of \ac{pima} is $10^{-2}$ if $\tau_{\rm B}=D_{\rm B}\approx19\cdot10^3$ \ac{usd} for $\Lambda_{\rm B} = 600$, and if $\tau_{\rm B}\approx85\cdot10^3$ \ac{usd} for $\Lambda_{\rm B} = 3000$. 
In this comparison, \ac{pima} is shown to be effective in minimizing burst transmission time, being very close to the \ac{cra}-2 solutions with $M_{\rm p} = 1000$ and $M_{\rm p} = \Lambda{\rm B}$ in the low traffic scenario, while outperforming both fixed-length preamble solutions for $\Lambda_{\rm B} = 3000$.

\section{Conclusions}\label{sec:conclusions}
We have proposed the \ac{pima} protocol, a semi-\ac{gf} coordinated multiple access scheme for short packet transmission, based on the knowledge of the number of users that have packets to transmit. \Ac{pima} organizes time into frames, and each frame includes a preliminary phase (the \ac{pia} sub-frame), where \ac{bs} estimates the number of active users, and a second phase (the \ac{dt} sub-frame), wherein the actual data transmissions are carried out.
We derive the optimal scheduling in the case of i.i.d. activations and assess its performance for different users' activation statistics. 
The numerical results obtained in such scenarios show that \ac{pima} is able to achieve extremely low latency with respect to state-of-the-art orthogonal multiple access solutions due to its low overhead and is able to adapt to different activation conditions by exploiting the partial knowledge of the instantaneous traffic load. 

\balance

\bibliographystyle{IEEEtran}
\bibliography{IEEEabrv, Bibliography.bib}

\end{document}